\documentclass[12pt,letterpaper]{article}
\usepackage{amssymb,amsmath,amsthm,amsxtra,amsfonts}
\usepackage{graphicx}
\usepackage{mathrsfs}  
\usepackage{enumerate}
\usepackage{microtype}
\usepackage{natbib}
\usepackage{url} 
\usepackage[usenames,dvipsnames,table]{xcolor}
\usepackage[colorlinks=true,citecolor=Blue,urlcolor=Blue,linkcolor=Blue]{hyperref}

\newtheorem{theorem}{Theorem}[section]
\newtheorem{proposition}{Proposition}[section]
\newtheorem{lemma}{Lemma}[section]

\newtheorem{assumption}{Assumption}[section]

\newcommand{\R}{\ensuremath{\mathbf{R}}}
\newcommand{\Z}{\ensuremath{\mathbf{Z}}}

\newcommand{\Nn}{\ensuremath{\mathbf{N}}}
\newcommand{\E}{\ensuremath{\mathbb{E}}}
\newcommand{\dx}{\ensuremath{\mathrm{d}}}
\newcommand{\Var}{\ensuremath{\mathrm{Var}}}
\newcommand{\Cov}{\ensuremath{\mathrm{Cov}}}
\newcommand{\si}{\perp \! \! \! \perp}

\usepackage{longtable}

\begin{document}

\def\spacingset#1{\renewcommand{\baselinestretch}%
	{#1}\small\normalsize} \spacingset{1}


	\title{\bf High-Dimensional Granger Causality Tests with an Application to VIX and News\footnote{We thank  participants at the Financial Econometrics Conference at the TSE Toulouse, the JRC Big Data and Forecasting Conference, the Big Data and Machine Learning in Econometrics, Finance, and Statistics Conference at the University of Chicago, the Nontraditional Data, Machine Learning, and Natural Language Processing in Macroeconomics Conference at the Board of Governors, the AI Innovations Forum organized by SAS and the Kenan Institute, the 12th World Congress of the Econometric Society, the 2021 Econometric Society Winter Meetings and seminar participants at Vanderbilt University for helpful comments.}}
	\author{Andrii Babii\thanks{Department of Economics, University of North Carolina--Chapel Hill - Gardner Hall, CB 3305
			Chapel Hill, NC 27599-3305. Email: babii.andrii@gmail.com} \and Eric Ghysels\thanks{Department of Economics and Kenan-Flagler Business School, University of North Carolina--Chapel Hill. Email: eghysels@unc.edu.} \and
		Jonas Striaukas\thanks{LIDAM UC Louvain and FRS--FNRS Research Fellow. Email: jonas.striaukas@gmail.com.}}
	\maketitle

\bigskip
\begin{abstract}
	We study Granger causality testing for high-dimensional time series using regularized regressions. To perform proper inference, we rely on heteroskedasticity and autocorrelation consistent (HAC) estimation of the asymptotic variance and develop the inferential theory in the high-dimensional setting. To recognize the time series data structures we focus on the sparse-group LASSO estimator, which includes the LASSO and the group LASSO as special cases. We establish the debiased central limit theorem for low dimensional groups of regression coefficients and study the HAC estimator of the long-run variance based on the sparse-group LASSO residuals. This leads to valid time series inference for individual regression coefficients as well as groups, including Granger causality tests. The treatment relies on a new Fuk-Nagaev inequality for a class of $\tau$-mixing processes with heavier than Gaussian tails, which is of independent interest. In an empirical application, we study the Granger causal relationship between the VIX and financial news.
\end{abstract}

\noindent%
\vfill
\thispagestyle{empty}
\setcounter{page}{0}
\newpage
\spacingset{1.5} 

\section{Introduction} 
\label{sec:intro}

Modern time series analysis is increasingly using  high-dimensional datasets, typically available at different frequencies. Conventional time series are often supplemented with non-traditional data, such as the high-dimensional data coming from the natural language processing. For instance, \cite{bybee2019structure} extract 180 topic attention series from the over 800,000 daily \textit{Wall Street Journal} news articles during 1984-2017 that have shown by  \cite{babii2020midasml} to be a useful supplement to more traditional macroeconomic and financial datasets for nowcasting US GDP growth. 

\smallskip

In his seminal paper, Clive Granger defined causality in terms of high-dimensional time series data. His formal definition, see \cite[Definition 1]{granger1969investigating}, considered all the information accumulated in the universe up to time $t - 1$ (a process he called $U_t$) and examined predictability using $U_t$ with and without a specific series of interest $Y_t.$ It is still an open question how to implement Granger's test in a high-dimensional time series setting. It is the purpose of this paper to do this via regularized regressions using HAC-based inference. In a sense, we are trying to implement Granger's original idea of causality.\footnote{There exists an extensive literature on causal inference with machine learning methods within the \textit{static} Neyman-Rubin's potential outcomes framework; see \cite{athey2019machine} for the excellent review and further references.} 

\smallskip

Following \cite{babii2020midasml}, we focus on the structured sparsity approach based on the sparse-group LASSO (sg-LASSO) regularization for the high-dimensional time series analysis. The sg-LASSO allows capturing the group structures present in high-dimensional time series regressions where a single covariate with its lags constitutes a group. Alternatively, we can also combine covariates of similar nature in groups. An attractive feature of this estimator is that it encompasses the LASSO and the group LASSO as special cases, hence, it allows improving upon the unstructured LASSO in the high-dimensional time-series setting. At the same time, the sg-LASSO can learn the distribution of time series lags in a data-driven way solving elegantly the model selection problem that dates back to \cite{fisher1937note}.\footnote{The distributed lag literature can be traced back to \cite{fisher1925our}; see also \cite{almon1965distributed}, \cite{sims1971discrete}, and \cite{shiller1973distributed}, as well as more recent mixed frequency data sampling (MIDAS) approach in \cite{ghysels2006predicting}, \cite{ghysels2007midas}, and \cite{andreou2013should}.}
In particular, the group structure can also accommodate data sampled at different frequencies as discussed in detail by \cite{babii2020midasml}.
\smallskip

The proper inference for time-series data relies on the heteroskedasticity and autocorrelation consistent (HAC) estimation of the long-run variance; see \cite{eicker1963asymptotic}, \cite{huber1967behavior}, \cite{white1980heteroskedasticity}, \cite{gallant1987nonlinear}, \cite{newey1987simple}, \cite{andrews1991heteroskedasticity}, among others.\footnote{For stationary time series, the HAC estimation of the long-run variance is the same problem as the estimation of the value of the spectral density at zero which itself has even longer history dating back to the smoothed periodogram estimators; see \cite{daniell1946discussion}, \cite{bartlett1948smoothing}, and \cite{parzen1957consistent}.} Despite the increasing popularity of the LASSO in finance and more generally in time series empirical research, to the best of our knowledge, the validity of HAC-based inference for LASSO has not been established in the relevant literature.\footnote{See \cite{chernozhukov2019lasso} for LASSO inference with dependent data and \cite{feng2019taming} for an asset pricing application; see also \cite{belloni2014inference} and \cite{van2014asymptotically} for i.i.d.\ data; and \cite{chiang2019lasso} for exchangeable arrays.} The HAC-based inference is robust to the model misspecification and leads to the valid Granger causality tests even when the fitted regression function has only projection interpretation which is the case for the projection-based definition of the Granger causality. Developing the asymptotic theory for the linear projection model with autoregressive lags and covariates, however, is challenging because the underlying processes are typically \textit{not} $\beta$-mixing.\footnote{More generally, it is known that the linear transformations based on infinitely many lags do not preserve the $\alpha$- or $\beta$-mixing property.}

\smallskip

In this paper, we obtain the debiased central limit theorem with explicit bias correction for the sg-LASSO estimator and time series data, which to the best of our knowledge is new. Next, we establish the formal statistical properties of the HAC estimator based on the sg-LASSO residuals in the high-dimensional environment when the number of covariates can increase faster than the sample size. We show that the convergence rate of the HAC estimator can be affected by the tails and the persistence of the data, which is a new phenomenon compared to low-dimensional regressions. For the practical implementation, this implies that the optimal choice of the bandwidth parameter for the HAC estimator should scale appropriately with the number of covariates, the tails, and the persistence of the data. These results allow us to perform inference for groups of coefficients, including the (mixed-frequency) Granger causality tests.

\smallskip

Our asymptotic theory applies to the heavy-tailed time series data, which is often observed in financial and economic applications. To that end, we establish a new Fuk-Nagaev inequality, see \cite{fuk1971probability}, for \textit{$\tau$-mixing} processes with \textit{polynomial tails}. The class of $\tau$-mixing processes is flexible enough for developing the asymptotic theory for the linear projection model and, at the same time, it contains the class of $\alpha$-mixing processes as a special case.

\smallskip

The paper is organized as follows. We start with the large sample approximation to the distribution of the sg-LASSO estimator (and as a consequence of the LASSO and the group LASSO) with $\tau$-mixing data in section \ref{sec:inference}. Next, we consider the HAC estimator of the asymptotic long-run variance based on the sg-LASSO residuals and study the inference for groups of regression coefficients. In section~\ref{sec:fn_inequality}, we establish a suitable version of the Fuk-Nagaev inequality for $\tau$-mixing processes. We report on a Monte Carlo study in section \ref{sec:mc} which provides further insights about the validity of our theoretical analysis in finite sample settings typically encountered in empirical applications. Section \ref{sec:application} covers an empirical application examining the Granger causal relations between the VIX and financial news. Conclusions appear in section \ref{sec:conclusion}. Proofs and supplementary results appear in the appendix and the supplementary material.

\paragraph{Notation:} For a random variable $X\in\R$ and $q\geq 1$, let $\|X\|_q=(\E|X|^q)^{1/q}$ be its $L_q$ norm. For $p\in\Nn$, put $[p] = \{1,2,\dots,p\}$. For a vector $\Delta\in\R^p$ and a subset $J\subset [p]$, let $\Delta_J$ be a vector in $\R^p$ with the same coordinates as $\Delta$ on $J$ and zero coordinates on $J^c$. Let $\mathcal{G}=\{G_g:\; g\geq 1 \}$ be a partition of $[p]$ defining groups. For a vector of regression coefficients $\beta\in\R^p$, the sparse-group structure is described by a pair $(S_0,\mathcal{G}_0)$, where $S_0=\{j\in[p]:\;\beta_j\ne 0 \}$ is the support of $\beta$ and  $\mathcal{G}_0 = \left\{G\in\mathcal{G}:\; \beta_{G} \ne 0\right\}$ is its group support. For $b\in\R^p$ and $q\geq 1$, its $\ell_q$ norm is denoted $|b|_q = \left(\sum_{j\in [p]}|b_j|^q\right)^{1/q}$ if $q<\infty$ and $|b|_\infty = \max_{j\in[p]}|b_j|$ if $q=\infty$. For $\mathbf{u},\mathbf{v}\in\R^T$, the empirical inner product is defined as $\langle \mathbf{u},\mathbf{v}\rangle_T = \frac{1}{T}\sum_{t=1}^T u_tv_t$ with the induced empirical norm $\|.\|_T^2=\langle.,.\rangle_T=|.|_2^2/T$. For a symmetric $p\times p$ matrix $A$, let $\mathrm{vech}(A)\in\R^{p(p+1)/2}$ be its vectorization consisting of the lower triangular and the diagonal part. Let $A_G$ be a sub-matrix consisting of rows of $A$ corresponding to indices in $G\subset[p]$. If $G=\{j\}$ for some $j\in[p]$, then we simply put $A_G=A_j$. For a $p\times p$ matrix $A$, let $\|A\|_\infty = \max_{j\in[p]}|A_j|_1$ be its matrix norm. For $a,b\in\R$, we put $a\vee b = \max\{a,b\}$ and $a\wedge b = \min\{a,b\}$. Lastly, we write $a_n\lesssim b_n$ if there exists a (sufficiently large) absolute constant $C$ such that $a_n\leq C b_n$ for all $n\geq 1$ and $a_n\sim b_n$ if $a_n\lesssim b_n$ and $b_n\lesssim a_n$.

\section{HAC-based inference for sg-LASSO}\label{sec:inference}
In this section, we cover the large sample approximation to the distribution of the sg-LASSO (LASSO and group LASSO) estimator for $\tau$-mixing processes. Next, we consider the HAC estimator of the asymptotic long-run variance based on the sg-LASSO residuals and consider the Granger causality tests. In the first subsection, we cover the debiased central limit theorem. The next subsection covers the HAC estimator and the final subsection pertains to the Granger causality tests.

\subsection{Debiased central limit theorem}
Consider a generic linear projection model
\begin{equation*}
	y_{t} = \sum_{j=1}^\infty\beta_j x_{t,j} + u_t,\qquad \E[u_tx_{t,j}]=0,\quad\forall j\geq 1,\qquad t\in\Z,
\end{equation*}
where $(y_t)_{t\in\Z}$ is a real-valued stochastic process and predictors may include the intercept, some covariates, (mixed-frequency) lags of covariates up to a certain order, as well as lags of the dependent variable. For a sample of size $T$, in the vector notation, we write
\begin{equation*}
	\mathbf{y} = \mathbf{m} + \mathbf{u},
\end{equation*}
where $\mathbf{y}=(y_1,\dots,y_T)^\top$, $\mathbf{m} = (m_1,\dots,m_T)^\top$ with $m_t = \sum_{j=1}^\infty\beta_j x_{t,j}$, and $\mathbf{u}=(u_1,\dots,u_T)^\top$. We approximate $m_t$ with $x_t^\top\beta = \sum_{j=1}^p\beta_jx_{t,j}$ and put $\mathbf{X}\beta$, where $\mathbf{X}$ is $T\times p$ design matrix and $\beta\in\R^p$ is the unknown projection parameter. This approximation can be constructed from lagged values of $y_t$, some covariates, as well as lagged values of covariates measured at a higher frequency, in which case, we obtain the autoregressive distributed lag mixed frequency data sampling model (ARDL-MIDAS) described as
\begin{equation*}
	\phi(L)y_t = \sum_{k=1}^K\psi(L^{1/m};\beta_k)x_{t,k} + u_t,
\end{equation*}
where $\phi(L) = I - \rho_1L - \rho_2L^2 - \dots - \rho_JL^J$ is a low frequency lag polynomial and the MIDAS part $\psi(L^{1/m};\beta_k)x_{t,k} = \frac{1}{m}\sum_{j=1}^m\beta_{k,j}x_{t-(j-1)/m,k}$ is a high-frequency lag polynomial; see \cite{andreou2013should} and \cite{babii2020midasml}. Note that when $m$ = 1 we have all data sampled at the same frequency and recover the standard autoregressive distributed lag (ARDL) model. The ARDL-MIDAS regression has a group structure where a single group is defined as all lags of $x_{t,k}$ or all lags of $y_t$ and following \cite{babii2020midasml}, we focus on the sparse-group LASSO (sg-LASSO) regularized estimator.\footnote{The sg-LASSO estimator allows selecting groups and important group members at the same time.} The leading example here is the MIDAS regression involving the projection of future low frequency series onto its own lags and lags of high frequency data aggregated via some dictionary, e.g., the set of Legendre polynomials. The setup also covers what is sometimes called the reverse MIDAS, see \cite{foroni2018using} and mixed frequency VAR, see \cite{ghysels2016macroeconomics}, involving the projection of high frequency data onto its own (high frequency) lags and low frequency data. Such regressions, which appear in the empirical application of the paper, simply amount to a different group structure.

\smallskip

The sg-LASSO, denoted $\hat\beta$, solves the regularized least-squares problem
\begin{equation}\label{eq:sgl}
	\min_{b\in\R^p}\|\mathbf{y} - \mathbf{X}b\|_T^2 + 2\lambda\Omega(b)
\end{equation}
with the regularization functional
\begin{equation*}
	\Omega(b) = \alpha|b|_1 + (1-\alpha)\|b\|_{2,1},
\end{equation*}
where $|b|_1 = \sum_{j=1}^p|b_j|$ is the $\ell_1$ norm corresponding to the LASSO penalty, $\|b\|_{2,1}=\sum_{G\in\mathcal{G}}|b_{G}|_2$ is the group LASSO penalty, and the group structure $\mathcal{G}$ is a partition of $[p]=\{1,2,\dots,p\}$ specified by the econometrician.

We measure the persistence of the series with $\tau$-mixing coefficients. For a $\sigma$-algebra $\mathcal{M}$ and a random vector $\xi\in\R^l$, put
\begin{equation*}
	\tau(\mathcal{M},\xi) = \left\|\sup_{f\in\mathrm{Lip}_1}|\E(f(\xi)|\mathcal{M}) - \E(f(\xi))|\right\|_1,
\end{equation*}
where $\mathrm{Lip}_1=\{f:\R^l\to\R:\; |f(x)-f(y)|\leq |x-y|_1 \}$ is a set of $1$-Lipschitz functions. Let $(\xi_t)_{t\in\Z}$ be a stochastic process and let $\mathcal{M}_t=\sigma(\xi_t,\xi_{t-1},\dots)$ be its natural filtration. The $\tau$-mixing coefficient is defined as
\begin{equation*}
	\tau_k = \sup_{j\geq 1}\frac{1}{j}\sup_{t+k\leq t_1<\dots<t_j}\tau(\mathcal{M}_t,(\xi_{t_1},\dots,\xi_{t_j})),\qquad k\geq 0,
\end{equation*}
where the supremum is taken over $t$ and $(t_1,\dots,t_j)$. The process is called $\tau$-mixing if $\tau_k\downarrow0$ as $k\uparrow\infty$; see Lemma~\ref{lemma:covariances} for the comparison of this coefficient to the mixingale and the $\alpha$-mixing coefficients. The following assumptions impose tail and moment conditions on the series of interest.
\begin{assumption}[Data]\label{as:data}
	The processes $(u_tx_{t})_{t\in\Z}$ and $(x_{t}x_{t}^\top)_{t\in\Z}$ are stationary for every $p\geq 1$ and such that (i) $\|u_t\|_q<\infty$ and $\max_{j\in[p]}\|x_{t,j}\|_r = O(1)$ for some $q>2r/(r-2)$ and $r>4$; (ii) for every $j,l\in[p]$, the $\tau$-mixing coefficients of $(u_tx_{t,j})_{t\in\Z}$ and $(x_{t,j}x_{t,l})$ are $\tau_k \leq ck^{-a}$ and $\tilde\tau_k\leq ck^{-b}$ for all $k\geq 0$ and some universal constants $c>0$, $a>(\varsigma-1)/(\varsigma-2)$, $b>(r-2)/(r-4)$, and $\varsigma=qr/(q+r)$.
\end{assumption}

\noindent Assumption~\ref{as:data} can be relaxed to non-stationary data with stable variances of partial sums at the cost of heavier notation. It allows for heavy-tailed and persistent data. For instance, it requires that either both covariates and the error process have at least $4+\epsilon$ finite moments, or that the error process has at least $2+\epsilon$ finite moments, whenever covariates are sufficiently integrable. It is also known that the $\tau$-mixing coefficients decline exponentially fast for geometrically ergodic Markov chains, including the stationary AR(1) process, so condition (ii) allows for relatively persistent data. Next, we require that the covariance matrix of covariates is invertible.
\begin{assumption}[Covariance]\label{as:covariance}
	There exists a universal constant $\gamma>0$ such that the smallest eigenvalue of $\Sigma=\E[x_tx_t^\top]$ is bounded away from zero by $\gamma$.
\end{assumption}
\noindent Assumption~\ref{as:covariance} ensures that the precision matrix $\Theta = \Sigma^{-1}$ exists and rules out perfect multicollinearity. It also requires that the smallest eigenvalue of $\Sigma$ is bounded away from zero by $\gamma$ independently of the dimension $p$ which is the case, e.g., for the spiked identity and the Toeplitz covariance structures. Strictly, speaking this condition can be relaxed to $\gamma\downarrow 0$ as $p\uparrow\infty$ at the cost of slower convergence rates and more involved conditions on rates, in which case $\gamma$ can be interpreted as a measure of ill-posedness; see \cite{carrasco2007linear}. The next assumption describes the rate of the regularization parameter, which is governed by the Fuk-Nagaev inequality; see Theorem~\ref{thm:tails_polynomial} and Eq.~\ref{eq:concentration}.

\begin{assumption}[Regularization]\label{as:tuning}
	For some $\delta\in(0,1)$
	\begin{equation*}
	\lambda \sim\left(\frac{p}{\delta T^{\kappa-1}}\right)^{1/\kappa}\vee\sqrt{\frac{\log(8p/\delta)}{T}},
	\end{equation*}
	where $\kappa = ((a+1)\varsigma-1)/(a+\varsigma-1)$, where $a,\varsigma$ are as in Assumption~\ref{as:data}.
\end{assumption}

\noindent Lastly, we impose the following condition on the misspecification error, the number of covariates $p$, the sparsity constant $s_\alpha$, and the sample size $T$. 
\begin{assumption}\label{as:rates}
	(i) $\|\mathbf{m} - \mathbf{X}\beta\|_T^2 = O_P\left(s_\alpha\lambda^2\right)$; (ii) $s_\alpha^\mu p^2T^{1-\mu}\to0$ and $p^2\exp(-cT/s_\alpha^2) \to 0$ as $T\to\infty$, where $s_\alpha$ is the effective sparsity of $\beta$ (defined below) and $\mu = ((b+1)r-2)/(r+2(b-1))$.
\end{assumption}
\noindent The effective sparsity constant $\sqrt{s_\alpha} = \alpha\sqrt{|S_0|} + (1-\alpha)\sqrt{|\mathcal{G}_0|}$ is a linear combination of the sparsity $|S_0|$ (number of non-zero coefficients) and the group sparsity $|\mathcal{G}_0|$ (number of active groups). It reflects the finite sample advantages of imposing the sparse-group structure as $|\mathcal{G}_0|$ can be significantly smaller than $|S_0|$ that appears in the theory of the standard LASSO estimator. Throughout the paper we assume that the groups have fixed size, which is well-justified in time-series applications of interest.

The four assumptions listed above are needed for the prediction and estimation consistency of the sg-LASSO estimator; see Theorem~\ref{thm:rates} in the supplementary material. Next, let $v_{t,j}$ be the regression error in $j^{th}$ nodewise LASSO regression; see the following subsection for more details. Put also $s=s_\alpha\vee S$, $S=\max_{j\in G}S_j$, where $S_j$ is the number of non-zero coefficients in the $j^{th}$ row of $\Theta$. The following assumption describes an additional set of sufficient conditions for the debiased central limit theorem.

\begin{assumption}\label{as:clt}
	(i) $\sup_x\E[u_t^2|x_t=x]=O(1)$; (ii) $\|\Theta_G\|_\infty = O(1)$ for some $G\subset[p]$ of fixed size; (iii) the long run variance of $(u_t^2)_{t\in\Z}$ and $(v_{t,j}^2)_{t\in\Z}$ exists for every $j\in G$; (iv) $s^2\log^2 p/T\to 0$ and $p/\sqrt{T^{\kappa-2}\log^{\kappa} p}\to 0$; (v) $\|\mathbf{m}-\mathbf{X}\beta\|_T = o_P(T^{-1/2})$; (vi) for every $j,l\in[p]$ and $k\geq 0$, the $\tau$-mixing coefficients of $(u_tu_{t+k}x_{t,j}x_{t+k,l})_{t\in\Z}$ are $\check\tau_t\leq ct^{-d}$ for some universal constants $c>0$ and $d>1$.
\end{assumption}
Assumption (i) requires that the conditional variance of the regression error is bounded. Condition (ii) requires that the rows of the precision matrix have bounded $\ell_1$ norm and is a plausible assumption in the high-dimensional setting, where the inverse covariance matrix is often sparse, e.g., in the Gaussian graphical model. Condition (iii) is a mild restriction needed for the consistency of the sample variance of regression errors. The sparsity condition $s^2\log^2p/T\to 0$ is also used in \cite{van2014asymptotically} who assume that the data are sub-Gaussian, see their Corollary 2.1. On the other hand, the rate condition on the dimension $p/\sqrt{T^{\kappa - 2}\log^\kappa p}\to 0$,  is additional condition needed in our setting when regression errors are not Gaussian and may only have a certain number of finite moments. Lastly, condition (v) is trivially satisfied when the projection coefficients are sparse and, more generally, it requires that the misspecification error vanishes asymptotically sufficiently fast. Conditions of this type are standard in nonparametric literature.

Let $B = \hat\Theta\mathbf{X}^\top(\mathbf{y} - \mathbf{X}\hat\beta) / T$ denote the bias-correction for the sg-LASSO estimator, where $\hat\Theta$ is the nodewise LASSO estimator of the precision matrix $\Theta$; see the following subsection for more details. The following result describes a large-sample approximation to the distribution of the debiased sg-LASSO estimator with serially correlated non-Gaussian regression errors.
\begin{theorem}\label{thm:clt}
	Suppose that Assumptions~\ref{as:data}, \ref{as:covariance}, \ref{as:tuning}, \ref{as:rates}, and \ref{as:clt} are satisfied for the sg-LASSO regression and for each nodewise LASSO regression $j\in G$. Then 
	\begin{equation*}
		\sqrt{T}(\hat\beta_G + B_G - \beta_G) \xrightarrow{d} N(0,\Xi_G)
	\end{equation*}
	with the long-run variance\footnote{With slight abuse of notation we use $\beta_G\in\R^{|G|}$ to denote the subvector of elements of $\beta\in\R^p$ indexed by $G$.} $\Xi_G = \lim_{T\to\infty}\Var\left(\frac{1}{\sqrt{T}}\sum_{t=1}^T u_t\Theta_GX_t\right)$.
\end{theorem}
It is worth mentioning that since the group $G$ has fixed size and the rows of $\Theta$ have finite $\ell_1$ norm, the long-run variance $\Xi_G$ exists under the maintained assumptions; see Proposition~\ref{prop:long_run} in the Appendix for a precise statement of this result.

Theorem~\ref{thm:clt} extends \cite{van2014asymptotically} to non-Gaussian, heavy-tailed and persistent time series data and describes the long run asymptotic variance for the low-dimensional group of regression coefficients estimated with the sg-LASSO. One could also consider Gaussian approximations for groups of increasing size, which requires an appropriate high-dimensional Gaussian approximation result for $\tau$-mixing processes and is left for future research; see \cite{chernozhukov2013gaussian} for a comprehensive review of related coupling results in the i.i.d.\ case.

\subsection{Nodewise LASSO}
The bias-correction term $B$ and the expression of the long-run variance in Theorem~\ref{thm:clt} depend on the appropriate estimator of the precision matrix $\Theta = \Sigma^{-1}$. We focus on the nodewise LASSO estimator of $\Theta$, following the work of \cite{meinshausen2006high} and \cite{van2014asymptotically}. The estimator is based on the  observation that the covariance matrix of the partitioned vector $X = (X_j,X_{-j}^\top)^\top\in\R\times\R^{p-1}$  can be written as
\begin{equation*}
	\Sigma = \E[XX^\top] = \begin{pmatrix}
		\Sigma_{j,j} & \Sigma_{j,-j} \\
		\Sigma_{-j,j} & \Sigma_{-j,-j}
	\end{pmatrix},
\end{equation*}
where $\Sigma_{j,j}=\E[X_j^2]$ and all other elements similarly defined. By the partitioned inverse formula, the $1^{st}$ row of the precision matrix $\Theta=\Sigma^{-1}$ is
\begin{equation*}
	\Theta_j = \sigma_j^{-2}\begin{pmatrix}
		1 & -\gamma_j^\top
	\end{pmatrix},
\end{equation*}
where $\gamma_j = \Sigma_{-j,-j}^{-1}\Sigma_{-j,j}$ is the projection coefficient in the regression of $X_j$ on $X_{-j}$
\begin{equation}\label{eq:nodewise_regressions}
	X_j = X_{-j}^\top\gamma_j + v_j,\qquad\E[X_{-j}v_j] = 0,
\end{equation}
and $\sigma_j^2 = \Sigma_{j,j} - \Sigma_{j,-j}\gamma_j = \E[v_j^2]$ is the variance of the projection error.\footnote{To ensure that the projection coefficient is well defined and does not change with the dimension of the model $p$, we can consider the limiting linear projection model and take into account the approximation error.} This suggests estimating the $1^{st}$ row of the precision matrix as $\hat\Theta_j = \hat\sigma_j^{-2}\begin{pmatrix}
1 & -\hat\gamma_j^\top
\end{pmatrix}$ with $\hat\gamma_j$ solving
\begin{equation*}
	\min_{\gamma\in\R^{p-1}}\|\mathbf{X}_j - \mathbf{X}_{-j}\gamma\|^2_T + 2\lambda_j|\gamma|_1
\end{equation*}
and
\begin{equation*}
	\hat\sigma_j^2 = \|\mathbf{X}_j - \mathbf{X}_{-j}\hat \gamma_j\|^2_T + \lambda_j|\hat\gamma_j|,
\end{equation*}
where $\mathbf{X}_j\in\R^T$ is the column vector of observations of $x_j\in\R$ and $\mathbf{X}_{-j}$ is the $T\times (p-1)$ matrix of observations of $x_{-j}\in\R^{p-1}$. In the matrix notation, the nodewise LASSO estimator of $\Theta$ can be written then as $\hat\Theta = \hat B^{-1}\hat C$ with
\begin{equation*}
	\hat C = \begin{pmatrix}
		1 & -\hat\gamma_{1,1} & \dots & -\hat\gamma_{1,p-1} \\
		-\hat\gamma_{2,1} & 1 & \dots & -\hat\gamma_{2,p-1} \\
		\vdots & \vdots & \ddots & \vdots \\
		-\hat\gamma_{p-1,1} & \dots & -\hat\gamma_{p-1,p-1} & 1
	\end{pmatrix}\qquad \text{and}\qquad \hat B = \mathrm{diag}(\hat\sigma_1^2,\dots,\hat\sigma_p^2).
\end{equation*}

\subsection{HAC estimator}
Next, we focus on the HAC estimator based on sg-LASSO residuals, covering the LASSO and the group LASSO as special cases. For a group $G\subset [p]$ of a fixed size, the HAC estimator of the long-run variance is
\begin{equation}
	\label{eq:hacformula}
	\hat\Xi_G = \sum_{|k|<T}K\left(\frac{k}{M_T}\right)\hat{\Gamma}_k,
\end{equation}
where $\hat{\Gamma}_k = \hat\Theta_G\left(\frac{1}{T}\sum_{t=1}^{T-k}\hat u_t\hat u_{t+k} x_tx_{t+k}^\top\right)\hat\Theta_G^\top$, $\hat u_t$ is the sg-LASSO residual, and $\hat\Gamma_{-k}=\hat{\Gamma}_k^\top$. The kernel function $K:\R\to[-1,1]$ with $K(0)=1$ is puts less weight on more distant noisy covariances, while $M_T\uparrow\infty$ is a bandwidth (or lag truncation) parameter, see \cite{parzen1957consistent} and \cite{andrews1991heteroskedasticity}. Several choices of the kernel function are possible, for example, the Parzen kernel is
\begin{align*}
	&  K_{PR}(x) = \begin{cases}		1-6x^2+6|x|^3 &\text{ for } 0 \leq |x| \leq 1/2, \\ 		2(1-|x|)^3 &\text{ for } 1/2 \leq |x| \leq 1, \\ 
		0&\text{ otherwise}. 
	\end{cases}
\end{align*}
It is worth recalling that the Parzen and the Quadratic spectral kernels are high-order kernels that superior to the Bartlett kernel used in \cite{newey1987simple}; see appendix for more details on the choice of the kernel. 

Note that under stationarity, the long-run variance in Theorem~\ref{thm:clt} simplifies to
\begin{equation*}
	\Xi_G = \sum_{k\in\Z}\Gamma_k,
\end{equation*}
where $\Gamma_k = \Theta_G\E[u_tx_tu_{t+k}x_{t+k}^\top]\Theta_G^\top$ and $\Gamma_{-k}=\Gamma^\top_k$. The following result characterizes the convergence rate of the HAC estimator pertaining to a group of regression coefficients $G\subset[p]$ based on the sg-LASSO residuals.
\begin{theorem}\label{prop:hac_residuals}
	Suppose that Assumptions~\ref{as:data}, \ref{as:covariance}, \ref{as:tuning}, \ref{as:rates}, \ref{as:clt} are satisfied for the sg-LASSO regression and for each nodewise LASSO regression $j\in G$. Suppose also that Assumptions~\ref{as:hac_moments}, and \ref{as:HAC} in the Appendix are satisfied for $V_t = (u_tv_{t,j}/\sigma^2_j)_{j\in G}$, $\kappa\geq \tilde q$ and that $s^\kappa pT^{1-4\kappa/5}\to 0$ as $M_T\to\infty$ and $T\to\infty$. Then
	\begin{equation*}
		\|\hat\Xi_G - \Xi_G\| = O_P\left(M_T\left(\frac{s p^{1/\kappa}}{T^{1-1/\kappa}} \vee s\sqrt{\frac{\log p}{T}} + \frac{s^2p^{2/\kappa}}{T^{2-3/\kappa}} + \frac{s^3p^{5/\kappa}}{T^{4-5/\kappa}}\right) + M_T^{-\varsigma} + T^{-(\varsigma\wedge 1)}\right).
	\end{equation*}
\end{theorem}
\noindent The first term in the inner parentheses is of the same order as the estimation error of the maximum between the estimation errors of the sg-LASSO and the nodewise LASSO. Theorem~\ref{prop:hac_residuals} suggests that the optimal choice of the bandwidth parameter should scale appropriately with the number of covariates $p$, the sparsity constant $s$, and the dependence-tails exponent $\kappa$.\footnote{A comprehensive study of the optimal bandwidth choice based on higher-order asymptotic expansions is beyond the scope of this paper and is left for future research, see, e.g., \cite{lazarus2018har} for the recent literature review and practical recommendations in the low-dimensional case.} This contrasts sharply with the HAC theory for regressions without regularization developed in \cite{andrews1991heteroskedasticity}, see also \cite{li2018uniform}, and allows for faster convergence rates of the HAC estimator.

\subsection{High-dimensional Granger causality tests}
Consider a linear projection model
\begin{equation*}
	y_{t+h} = \sum_{j\in G}\beta_jx_{t,j} + \sum_{j\in G^c} \beta_jx_{t,j} + u_t,\qquad \E[u_tx_{t,j}] = 0,\qquad \forall j\geq 1,
\end{equation*}
where $h\geq 0$ is the horizon, $G\subset[p]$ is a group of regression coefficients of interest, $x_t=\{x_{t,j}:\; j\in G \}$ represents the series for which we wish to test the Granger causality, and $\{x_{t,j}:j\in G^c \}$ represents all the remaining information available at time $t$. For instance, $x_t$ may contain $L$ low-frequency lags of some series $(z_t)_{t\in\Z}$, in which case $x_t = (z_t,z_{t-1},z_{t-2},\dots, z_{t-L})^\top$. Alternatively, it may contain low and/or high-frequency lags of $(z_t)_{t\in\Z}$ aggregated with dictionaries, e.g., Legendre polynomials as in \cite{babii2020midasml}. In both cases the dimensionality of $x_t$ small. On the other hand, the set of controls representing all the information available at time $t$ is high-dimensional. The Granger causality test corresponds to the following hypotheses
\begin{equation*}
	H_0:\; R\beta_{G} = 0\qquad \text{against} \qquad H_1:\; R\beta_{G}\ne 0,
\end{equation*}
where $\beta_G = \{\beta_j:\; j\in G \}$.

It is worth mentioning our framework is based on the weakest notion of the Granger causality corresponding to the marginal improvement in time series projections due to the information contained in $x_t$. A stronger notion of Granger non-causality corresponds to the independence on the conditional mean of $y_t$ of $x_t$. Yet, even stronger version of Granger non-causality requires that $y_t|x_t,\sigma(x_{t,j}:j\in G^c) =_d y_t|\sigma(x_{t,j}:j\in G^c)$, so that the entire conditional distribution of $y_t$ does not depend on $x_t$; see also \cite{florens1982note}.

Let $R$ be $r\times |G|$ matrix of linear restrictions imposed on $\beta_G$. For the Granger causality test, we set $R=I_{|G|}$, but more generally, we might be interested in testing other linear restrictions implied by the economic theory. Assuming that $R$ is a full row rank matrix, consider the debiased Wald statistics
\begin{equation*}
	W_T = T\left[R(\hat\beta_{G} + B_{G} - \beta_{G})\right]^\top \left(R\hat\Xi_GR^\top\right)^+\left[R(\hat\beta_{G} + B_{G} - \beta_{G})\right],
\end{equation*}
where $A^+$ is the generalized inverse of $A$. It follows from Theorems~\ref{thm:clt} and \ref{prop:hac_residuals} that under $H_0$, $W_T\xrightarrow{d}\chi^2_r$. The Wald test rejects when $W_T>q_{1-\alpha}$, where $q_{1-\alpha}$ is the quantile of order $1-\alpha$ of $\chi^2_{r}$. More generally, the linear restrictions can be extended to the nonlinear restrictions by the usual Delta method argument.

\smallskip

For testing hypotheses on the increasing set of regression coefficients, it might be preferable to use the non-pivotal sup-norm based statistics, see \cite{ghysels2018testing}, due to the remarkable deoendence on the dimension in the high-dimensional setting; see \cite{chernozhukov2013gaussian} for high-dimensional Gaussian approximations with i.i.d.\ data.

\section{Fuk-Nagaev inequality}\label{sec:fn_inequality}
In this section, we describe a suitable for us version of the Fuk-Nagaev concentration inequality for the maximum of high-dimensional sums. The inequality allows for the data with polynomial tails and $\tau$-mixing coefficients decreasing at a polynomial rate. The following result does not require that the series is stationary.

\begin{theorem}\label{thm:tails_polynomial}
	Let $(\xi_{t})_{t\in\Z}$ be a centered stochastic process in $\R^p$ such that (i) for some $q>2$, $\max_{j\in[p],t\in[T]}\|\xi_{t,j}\|_q=O(1)$; (ii) for every $j\in[p]$, $\tau$-mixing coefficients of $\xi_{t,j}$ satisfy $\tau_k^{(j)} \leq ck^{-a}$ for some universal constants $a,c>0$. Then there exist $c_1,c_2>0$ such that for every $u>0$
	\begin{equation*}
	\Pr\left(\left|\sum_{t=1}^T\xi_t\right|_\infty > u\right) \leq c_1pTu^{-\kappa} + 4p\exp\left(-\frac{c_2u^2}{B_T^2}\right),
	\end{equation*}
	where\footnote{It is worth mentioning that the notation in this section is specific to generic stochastic processes and is independent from the rest of the paper. Thus $B_T$ here denotes the variance of partial sums and not the bias correction term of the LASSO estimator.} $\kappa = ((a+1)q-1)/(a+q-1)$, $B^2_T = \max_{j\in[p]}\sum_{t=1}^T\sum_{k=1}^T|\Cov(\xi_{t,j},\xi_{k,j})|$.
\end{theorem}
The inequality describes the mixture of the polynomial and Gaussian tails for the maximum of high-dimensional sums. In the limiting case of the i.i.d.\ data, as $a\to \infty$, the dependence-tails exponent $\kappa$ approaches $q$ and we recover the inequality for the independent data stated in \cite{fuk1971probability}, Corollary 4 for $p=1$. In this sense, the inequality in Theorem~\ref{thm:tails_polynomial} is sharp. It is well-known that the Fuk-Nagaev inequality delivers sharper estimates of tail probabilities in contrast to Markov's bound in conjunction with Rosenthal's moment inequality, cf. \cite{nagaev1998some}. The proof relies on the blocking technique, see \cite{bosq1993bernstein}, and the coupling inequality for $\tau$-mixing sequences, see \cite{dedecker2004coupling}, Lemma 5. In contrast to previous results, e.g., \cite{dedecker2004coupling}, Theorem 2, the inequality reflects the mixture of the polynomial and the exponential tails.

For stationary processes, by Lemma~\ref{lemma:B_T_bound} in the appendix, $B_T^2 = O(T)$ as long as $a>(q-1)/(q-2)$, whence we obtain from Theorem~\ref{thm:tails_polynomial} that for every $\delta\in(0,1)$
\begin{equation}\label{eq:concentration}
	\Pr\left(\left|\frac{1}{T}\sum_{t=1}^T\xi_t\right|_\infty \leq C\left(\frac{p}{\delta T^{\kappa-1}}\right)^{1/\kappa}\vee\sqrt{\frac{\log(8p/\delta)}{T}} \right) \geq 1 - \delta,
\end{equation}
where $C>0$ is some finite universal constant.

\section{Monte Carlo experiments \label{sec:mc}}
In this section, we aim to assess the debiased HAC-based inferences for the low-dimensional parameter in a high dimensional data setting. To that end, we draw covariates $\{x_{t,j},j\in[p]\}$ independently from the AR(1) process
\begin{equation*}
	x_{t,j} =  \rho x_{t-1,j} + \epsilon_{t,j}.
\end{equation*}
The regression error follows the AR(1) process
\begin{equation*}
	u_t =  \rho u_{t-1} + \nu_t,
\end{equation*}
where errors are $\epsilon_{t,j},\nu_t  \sim_{i.i.d.} N(0,1)$. The vector of population regression coefficients $\beta$ has the first five non-zero entries which are drawn from Uniform$(0,4)$ and all remaining entries are zero. The sample size is $T\in\{100,1000\}$ and the number of covariates is $p\in\{10,200\}$. We set the persistence parameter $\rho = 0.6$ and focus on the LASSO estimator to estimate coefficients $\hat \beta.$ Throughout the experiment, we choose the LASSO tuning parameters using the 10-fold cross-validation, defining folds as adjacent over time blocks. 

\smallskip

We report the average coverage (av.\ cov) and the average length of confidence intervals for the nominal coverage of 0.95 and on a grid of values of the bandwidth parameter $M_T \in \{5, 10, \dots, 60\}$, using the Parzen kernel. We estimate the long run covariance matrix $\hat\Xi$ using the LASSO residuals, denoted $\hat u_t$. We also use the nodewise LASSO regressions to estimate the precision matrix $\Theta$. The first step is to compute scores $\hat V_t =  \hat u_t x_t$, where $\hat u_t = y_t - x_t^\top\hat \beta$, and $\hat\beta$ is the LASSO estimator. Then we compute the high-dimensional HAC estimator using the formuala in equation (\ref{eq:hacformula}). 
We compute the pivotal statistics for each MC experiment $i\in[N]$ and each coefficient $j\in[p]$ as
$\text{pivot}_{j}^{(i)}$ $\triangleq$	$(\hat \beta_{j}^{(i)} + B_{j}^{(i)} - \beta )/\sqrt{\hat\Xi_{j,j}^{(i)}/T}$,
where $B_{j}^{(i)} = \hat\Theta^{(i)}_j \mathbf{X}^{\top(i)} \hat{\mathbf{u}}^{(i)}/T$, and $\mathbf{\hat u}^{(i)}=\mathbf{y}^{(i)}-\mathbf{X}^{(i)}\hat\beta^{(i)}$. Then we compute the empirical coverage as
\begin{align*}
	\text{av.cov}_j = \frac{1}{N} \sum_{i=1}^{N}\mathbf{1}\{\text{pivot}_{j}^{i}  \in \left[-1.96, 1.96\right]\}
\end{align*}
and the average confidence interval length as
$\text{length}_{j}$ = $\frac{1}{N} \sum_{i=1}^{N} 2 \times 1.96 \times \sqrt{\hat\Xi_{j,j}^{(i)}/T}.$
The number of Monte Carlo experiments is set to $N$= $5000.$ 

\smallskip

We report average results over the active and inactive sets of the vector of coefficients. Table \ref{tab:infer_simul_gaussian} shows results for the small sample size ($T=100$) and the large sample size ($T=1000$). We find that the value of the bandwidth parameter $M_T$ should be smaller when the number of regressors $p$ is larger. For $T=100$ (small sample size), for the active set of coefficients, the best coverage is achieved when the bandwidth parameter is set at $10$ when $p=10$ and at $5$ when $p=200$. Results for the inactive set and $T=1000$ are similar. We also see that the increase in $p$ relative to $T$ leads to worse performance. Furthermore, the coverage improves when the bandwidth increases with the sample size. Lastly, as the sample size increases, the average coverage approaches the nominal level of 0.95 and the confidence interval shrinks in size. Overall, the simulation results confirm our theoretical findings.

\begin{table}[!htbp]
	\centering
	{\footnotesize
		\begin{tabular}{r cc cc c cc cc}
			&&\multicolumn{4}{c}{\underline{Average coverage (av. cov)}}
			&\multicolumn{4}{c}{\underline{Confidence interval length}}\\ 
			&\multicolumn{2}{c}{Active set of $\beta$}&\multicolumn{2}{c}{Inactive set of $\beta$}&&\multicolumn{2}{c}{Active set of $\beta$}&\multicolumn{2}{c}{Inactive set of $\beta$}\\\hline
			$M_T\backslash p$& 10 & 200 & 10 & 200 &  & 10 & 200 & 10 & 200 \\ 
			&\multicolumn{9}{c}{\underline{T=100}}\\
			5 & 0.830 & 0.755 & 0.827 & 0.747 &  & 0.286 & 0.382 & 0.277 & 0.368 \\ 
			10 & 0.834 & 0.750 & 0.835 & 0.746 &  & 0.305 & 0.401 & 0.291 & 0.376 \\ 
			15 & 0.824 & 0.753 & 0.832 & 0.745 &  & 0.314 & 0.411 & 0.294 & 0.376 \\ 
			20 & 0.821 & 0.735 & 0.826 & 0.743 &  & 0.320 & 0.420 & 0.296 & 0.374 \\ 
			25 & 0.816 & 0.739 & 0.820 & 0.741 &  & 0.325 & 0.427 & 0.297 & 0.372 \\ 
			30 & 0.807 & 0.738 & 0.814 & 0.740 &  & 0.329 & 0.434 & 0.297 & 0.370 \\ 
			35 & 0.806 & 0.733 & 0.807 & 0.738 &  & 0.333 & 0.441 & 0.297 & 0.368 \\ 
			40 & 0.801 & 0.737 & 0.802 & 0.735 &  & 0.337 & 0.447 & 0.297 & 0.366 \\ 
			&\multicolumn{9}{c}{\underline{T=1000}}\\
			5 & 0.913 & 0.848 & 0.913 & 0.879 &  & 0.081 & 0.067 & 0.081 & 0.067 \\ 
			10 & 0.932 & 0.865 & 0.932 & 0.894 &  & 0.087 & 0.070 & 0.087 & 0.070 \\ 
			15 & 0.935 & 0.866 & 0.936 & 0.894 &  & 0.088 & 0.070 & 0.088 & 0.070 \\ 
			20 & 0.936 & 0.868 & 0.936 & 0.897 &  & 0.089 & 0.071 & 0.088 & 0.071 \\ 
			25 & 0.936 & 0.867 & 0.936 & 0.895 &  & 0.089 & 0.071 & 0.089 & 0.071 \\ 
			30 & 0.937 & 0.866 & 0.937 & 0.895 &  & 0.089 & 0.071 & 0.089 & 0.070 \\ 
			35 & 0.936 & 0.866 & 0.936 & 0.895 &  & 0.089 & 0.071 & 0.089 & 0.070 \\ 
			40 & 0.934 & 0.865 & 0.934 & 0.894 &  & 0.089 & 0.071 & 0.088 & 0.070 \\  \hline
		\end{tabular}
		\caption{\small HAC-based inference simulation results -- The table reports average coverage (first four columns) and average length of confidence intervals (last four columns) for active and inactive sets of $\beta$ and for $T=100$ and $T=1000$. We report results for a set of bandwidth parameter $M_T$ values. 	\label{tab:infer_simul_gaussian}}
	} 
\end{table}

\clearpage 

\section{Testing Granger causality for VIX and financial news}\label{sec:application}

The CBOE Volatility Index, known as the VIX, is a popular measure of market-based expectation of future volatility and is often referred to as the ``fear index''.  The VIX index quotes the expected annualized change in the S\&P 500 index over the following 30 days, as computed from options-based theory and current options-market data. 

\smallskip

There is a large literature studying the theoretical and empirical properties of the VIX and it is impossible to cite only a few papers to do justice to all the outstanding research output on the topic. Focusing on Granger causal patterns, there are several studies pertaining to causality between the VIX and VIX futures. For example, \cite{bollen2017tail} suggest that the VIX futures lagged the VIX in the first few years after its introduction, and show an increasing dominance of VIX futures over time. Along similar lines, \cite{shu2012causality} study price‐discovery between VIX futures the spot VIX index and find evidence of a bi-directional causal pattern.

\smallskip

We study the causal relationship between financial news and the VIX. There is also a substantial literature on the impact of news releases on financial markets (see, e.g., \cite{andersen2003micro}). Traditionally, such analysis looks at news releases and studies the behavior of asset prices pre- and post-release. News is usually quantified numerically via the surprise component measured as the difference between an expectation prior to the release and the announcement. In the age of machine learning, the characterization of news has been expanded into the textual analysis of news coverage. To paraphrase the title of \cite{gentzkow2019text}, the text is treated as data. It is in this spirit that we conduct our high-dimensional Granger causality analysis between the VIX and news.

\smallskip

We use a data set from \cite{bybee2019structure} which contains 180 news attention monthly series, all of which potentially Granger cause future US equity market volatility.\footnote{We downloaded daily VIX data from St. Louis Fed FRED database and took the end-of-month values. The FRED mnemonic for the VIX is VIXCLS. Table with the full list of series appears in Appendix \ref{tab:selected}.} We estimate the following time series regression model

\begin{equation*}
\label{eq:ardlmidas}
y_{t+1} = \psi(L^{1/m}; \beta)y_t + \sum_{k=1}^K\rho_k x_{t,k} + u_t,\qquad t\in[T],
\end{equation*}
where $y_{t+1}$ is the value of the VIX at the end of month $t+1$, $\psi(L^{1/m}; \beta)y_{t}$ is a MIDAS polynomial of 22 daily VIX lags where the first lag is the last day of the month $t$, and $x_{t,k}$ is the $k$-th news attention series. Note that we only take one lag for the news attention series to simplify the model (and also an empirically justified simplification). The MIDAS polynomial of daily lags of the VIX involves Legendre polynomials of degree 3. Note that the specification is what is sometimes called a reverse MIDAS regression as mentioned earlier in the paper. Prior to estimating the regression model, we time demean the response and covariates such that the intercept is zero. We further standardize all covariates to have a unit standard deviation. The daily VIX lags are standardized before the aggregation. 

\smallskip

We apply the sg-LASSO estimator to estimate the slope coefficients and nodewise LASSO regressions to estimate the precision matrix. To fully exploit the group sparsity of sg-LASSO, we group all high-frequency lags of daily VIX, see \cite{babii2020midasml} for further details on such grouping. The news attention series are monthly and we are interested in whether the most recent news Granger causes the VIX, hence we don't apply the group structure along the time dimension. Instead, we group news attention series based on a broader theme that is available for each series, see \cite{bybee2019structure} for further details. Namely, the data set contains 24 broader topics which group each of the 180 news attention series.

\subsection{Main results}

We report the p-values for a range of $M_T$ values for series that appear to be significant at the 1\% or 5\% significance level for all $M_T\in\{20,40,60\}$ values and for two kernel functions, namely Parzen and Quadratic Spectral. The sample starts January 1990 January and ends June 2017, determined by the availability of the textual analysis data. Both sg-LASSO and LASSO tuning parameters are selected via 10-fold cross-validation, defining folds as adjacent blocks over the time dimension to take into account the
time series nature of the data. Similarly, we tune nodewise LASSO regressions for the precision matrix estimation.

\subsubsection{Granger causality of news topics}

The results appear in Table \ref{tab:vix_news_hac} which contains two main row blocks reporting results for the structured sg-LASSO and unstructured LASSO estimators, and two-column blocks, reporting results for two kernel functions. Irrespective of the initial estimator and kernel function, the lagged daily VIX and the Financial crisis news series are highly significant at 1\% significance level. 
Comparing results for the initial estimator, in the case of LASSO we see more significant predictors than for the sg-LASSO case, while the subset of significant covariates using sg-LASSO is a subset of the LASSO significant predictors. Many more series are selected by the initial LASSO estimator compared to the sg-LASSO, see Table \ref{tab:selected}. This suggests that relevant group structures are important, and may help in recovering salient relationships in the data. 

\begin{table}
	\centering
	\begin{tabular}{r  rrrrrrr}
		Variable $\backslash M_T$ & 20 & 40 & 60 &  & 20 & 40 & 60  \\ \hline
		&\multicolumn{7}{c}{\underline{sg-LASSO}} \\
		&\multicolumn{3}{c}{\underline{Parzen}}&&\multicolumn{3}{c}{\underline{Quadratic Spectral}} \\
		&\multicolumn{7}{c}{1\% significance} \\
		Daily VIX lags & 0.000 & 0.000 & 0.001 && 0.000 & 0.001 & 0.002 \\ 
		Financial crisis & 0.005 & 0.002 & 0.001 && 0.002 & 0.001 & 0.000 \\ 
		&\multicolumn{7}{c}{5\% significance} \\
		Aerospace/defense & 0.014 & 0.014 & 0.017 && 0.012 & 0.018 & 0.027 \\ 
		Recession & 0.011 & 0.008 & 0.009 && 0.008 & 0.008 & 0.013 \\
		&\multicolumn{7}{c}{\underline{LASSO}} \\
		&\multicolumn{3}{c}{\underline{Parzen}}&&\multicolumn{3}{c}{\underline{Quadratic Spectral}} \\
		&\multicolumn{7}{c}{1\% significance} \\
		Daily VIX lags & 0.000 & 0.000 & 0.000 && 0.000 & 0.000 & 0.000 \\ 
		Financial crisis & 0.001 & 0.000 & 0.000 && 0.000 & 0.000 & 0.000 \\ 
		Recession & 0.003 & 0.002 & 0.002 && 0.002 & 0.002 & 0.004 \\ 
		Marketing & 0.001 & 0.001 & 0.000 && 0.001 & 0.000 & 0.000 \\ 
		&\multicolumn{7}{c}{5\% significance} \\
		Aerospace/defense & 0.007 & 0.006 & 0.004 && 0.006 & 0.004 & 0.002 \\ 
		NY politics & 0.012 & 0.015 & 0.013 && 0.016 & 0.012 & 0.008\\ 
		Acquired investment banks & 0.018 & 0.006 & 0.002 && 0.007 & 0.001 & 0.000 \\ 
		\hline
	\end{tabular}
	\caption{VIX Granger causality results. We report p-values of series that are significant at 1\% and 5\% significance level for a range of $M_T$ values and both kernel functions. } 
	\label{tab:vix_news_hac}
\end{table}

\subsubsection{Bi-directional Granger causality}
We also test whether the daily VIX Granger causes Financial crisis news series. For this we run the following MIDAS regression model
\begin{equation*}
\label{eq:reverse_causality}
x_{t+1,j} = \psi(L^{1/m}; \beta)y_t + \sum_{k=1}^K\rho_k x_{t,k} + u_t,\qquad t\in[T],
\end{equation*}
where $x_{t+1,j}$ is the Financial crisis news series. We test whether daily VIX Granger causes future values of Financial crisis news series. Note that we only need to estimate the initial initial coefficient vector, since the precision matrix remains the same. The results appear in Table \ref{tab:reverse_causality}. They show a rather weak predictability of future news series by daily VIX suggesting a unidirectional Granger causality pattern. 
\begin{table}
	\centering
	\begin{tabular}{r  rrrrrrr}
		Variable $\backslash M_T$ & 20 & 40 & 60 &  & 20 & 40 & 60  \\ \hline
		&\multicolumn{7}{c}{\underline{sg-LASSO}} \\
		&\multicolumn{3}{c}{\underline{Parzen}}&&\multicolumn{3}{c}{\underline{Quadratic Spectral}} \\
		Daily VIX lags & 0.050 & 0.071 & 0.091 && 0.060 & 0.086 & 0.129 \\ 
		\hline
	\end{tabular}
	\caption{Bi-directional Granger causality results. We report p-values for a range of $M_T$ values and both kernel functions. } 
	\label{tab:reverse_causality}
\end{table}

\subsubsection{Granger causal clusters of news topics}

The news attention series are classified into 24 broader meta topics that group the individual news series according to a common theme. We test which group of individual news series Granger causes future VIX values. The results are reported in Table \ref{tab:group_causality}. They show that the group {\it Banks} is significant at 5\% significance level. This group consists of news series pertaining to news about Mortgages, Bank loans, Credit ratings, Nonperforming loans, Savings \& loans, and the Financial crisis.

\begin{table}
	\centering
	\begin{tabular}{r  rrrrrrr}
		Variable $\backslash M_T$ & 20 & 40 & 60 &  & 20 & 40 & 60  \\ \hline
		&\multicolumn{7}{c}{\underline{sg-LASSO}} \\
		&\multicolumn{3}{c}{\underline{Parzen}}&&\multicolumn{3}{c}{\underline{Quadratic Spectral}} \\
		Banks & 0.032 & 0.024 & 0.008 && 0.023 & 0.001 & 0.000 \\ 
		\hline
	\end{tabular}
	\caption{Group Granger causality results. We report p-values for a range of $M_T$ values and both kernel functions. } 
	\label{tab:group_causality}
\end{table}

\section{Conclusion \label{sec:conclusion}}
This paper develops valid inferential methods for high-dimensional time series regressions estimated with the sparse-group LASSO (sg-LASSO) estimator that encompasses the LASSO and the group LASSO as special cases. We derive the debiased central limit theorem with the explicit bias correction for the sg-LASSO with serially correlated regression errors. Furthermore, we also study HAC estimators of the long-run variance for low dimensional groups of regression coefficients and characterize how the optimal bandwidth parameter should scale with the sample size, the temporal dependence, as well as tails of the data. These results lead to the valid t- and Wald tests for the low-dimensional subset of parameters, such as Granger causality tests. Our treatment relies on a new suitable variation of the Fuk-Nagaev inequality for $\tau$-mixing processes which allows us to handle the time series data with polynomial tails. An interesting avenue for future research is to study more carefully the problem of the optimal data-driven bandwidth choice based on higher-order asymptotic expansions, see, e.g., \cite{sun2008optimal} for steps in this direction in low dimensional settings.

\smallskip 

In an empirical application we use a high-dimensional news attention series to study causal patterns between the VIX, sometimes called the fear index, and financial news. We find that almost exclusively the topic of financial crisis exhibits unidirectional Granger causality for the VIX.

\newpage

	\bibliographystyle{econometrica}
\bibliography{midas_ml}
\newpage
\setcounter{page}{1}
\setcounter{section}{0}
\setcounter{equation}{0}
\setcounter{table}{0}
\setcounter{figure}{0}
\renewcommand{\theequation}{A.\arabic{equation}}
\renewcommand\thetable{A.\arabic{table}}
\renewcommand\thefigure{A.\arabic{figure}}
\renewcommand\thesection{A.\arabic{section}}
\renewcommand\thepage{Appendix - \arabic{page}}
\renewcommand\thetheorem{A.\arabic{theorem}}

\begin{center}
	{\LARGE\bf APPENDIX}
\end{center}

\section{Proofs}\label{appsec:proofs}
\begin{proof}[Proof of Theorem~\ref{thm:clt}]
	By Fermat's rule, the sg-LASSO satisfies
	\begin{equation*}
		\mathbf{X}^\top(\mathbf{X}\hat\beta - \mathbf{y})/T + \lambda z^* = 0
	\end{equation*}
	for some $z^*\in\partial\Omega(\hat\beta)$, where $\partial\Omega(\hat\beta)$ is the sub-differential of $b\mapsto\Omega(b)$ at $\hat\beta$. Rearranging this expression and multiplying by $\hat\Theta$
	\begin{equation*}
		\hat\beta - \beta + \hat\Theta\lambda z^* = \hat\Theta\mathbf{X}^\top\mathbf{u}/T + (I-\hat\Theta\hat\Sigma)(\hat\beta - \beta) + \hat\Theta\mathbf{X}^\top(\mathbf{m}-\mathbf{X}\beta)/T,
	\end{equation*}
	where we use $\mathbf{y} = \mathbf{m} + \mathbf{u}$. Plugging in $\lambda z^*$ and multiplying by $\sqrt{T}$
	\begin{equation*}
		\begin{aligned}
			\sqrt{T}(\hat\beta - \beta + B) & = \hat\Theta\mathbf{X}^\top\mathbf{u}/\sqrt{T} + \sqrt{T}(I-\hat\Theta\hat\Sigma)(\hat\beta - \beta) + \hat\Theta\mathbf{X}^\top(\mathbf{m}-\mathbf{X}\beta)/\sqrt{T} \\
			& = \frac{1}{\sqrt{T}}\sum_{t=1}^Tu_t\Theta x_t + \frac{1}{\sqrt{T}}\sum_{t=1}^Tu_t(\hat\Theta - \Theta)X_t + \sqrt{T}(I-\hat\Theta\hat\Sigma)(\hat\beta - \beta) \\
			& \qquad + \hat\Theta\mathbf{X}^\top(\mathbf{m}-\mathbf{X}\beta)/\sqrt{T}.
		\end{aligned}
	\end{equation*}
	Next, we look at coefficients corresponding to $G\subset [p]$
	\begin{equation*}
		\begin{aligned}
			\sqrt{T}(\hat\beta_G - \beta_G + B_G) & = \frac{1}{\sqrt{T}}\sum_{t=1}^Tu_t\Theta_G x_t + \frac{1}{\sqrt{T}}\sum_{t=1}^Tu_t(\hat\Theta_G - \Theta_G)x_t + \sqrt{T}(I-\hat\Theta\hat\Sigma)_G(\hat\beta - \beta) \\
			& \qquad + \hat\Theta_G\mathbf{X}^\top(\mathbf{m}-\mathbf{X}\beta)/\sqrt{T} \\
			& \triangleq I_T + II_T + III_T + IV_T.
		\end{aligned}
	\end{equation*}
	We will show that $I_T\xrightarrow{d}N(0,\Xi_G)$ by the triangular array CLT, see \cite{neumann2013central}, Theorem 2.1. To that end, by the Cr\'{a}mer-Wold theorem, it is sufficient to show that $z^\top I_T \xrightarrow{d} z^\top N(0,\Xi_G)$ for every $z\in\R^{|G|}$. Note that under Assumptions~\ref{as:data} and \ref{as:clt} (i)-(ii)
	\begin{equation*}
		\begin{aligned}
			\sum_{t=1}^T\E\left|\frac{z^\top\xi_t}{\sqrt{T}}\right|^2 & = \E|u_tz^\top\Theta_Gx_t|^2 \\
			& \leq Cz^\top\Theta_G\Sigma\Theta_G^\top z \\
			& = O(1).
		\end{aligned}
	\end{equation*}	
	Therefore, since $q>2r/(r-2)$, we have $\varsigma>2$, and for every $\epsilon>0$
	\begin{equation*}
		\begin{aligned}
			\sum_{t=1}^T\E\left[\left|\frac{z^\top\xi_t}{\sqrt{T}}\right|^2\mathbf{1}\left\{\left|z^\top\xi_t\right| > \epsilon\sqrt{T} \right\}\right] & \leq \frac{\E\left|z^\top\xi_t\right|^{\varsigma}}{(\epsilon\sqrt{T})^{\varsigma-2}} = o(1).
		\end{aligned}
	\end{equation*}
	Next, under Assumptions~\ref{as:data} and \ref{as:clt} (i)-(ii), the long run variance
	\begin{equation*}
		\lim_{T\to\infty}\Var\left(\frac{1}{\sqrt{T}}\sum_{t=1}^Tz^\top\xi_t\right) = z^\top\Xi_Gz
	\end{equation*}
	exists by Proposition~\ref{prop:long_run}.
	
	Next, put $\mathcal{M} = \sigma(\xi_0,\xi_{-1},\xi_{-2},\dots)$, $Y = g(z^\top\xi_{t_1-t_h}/\sqrt{T},\dots,z^\top \xi_0/\sqrt{T})z^\top\xi_0$, and $X = z^\top\xi_r$ for some $t_h\geq 0$. Note that $X$ and $|XY|$ are integrable and that $Y$ is $\mathcal{M}$-measurable. Therefore, for every measurable function $g:\R^h\to\R$ with $\sup_x|g(x)|\leq 1$, by \cite{dedecker2003new}, Proposition 1, for all $h\in\Nn$ and all indices $1\leq t_1<t_2<\dots<t_h<t_h+r\leq t_h+s\leq T$
	\begin{equation*}
		\begin{aligned}
		& \left|\Cov\left(g(z^\top\xi_{t_1}/\sqrt{T},\dots,z^\top\xi_{t_h}/\sqrt{T})z^\top\xi_{t_h}/\sqrt{T},z^\top\xi_{t_h+r}/\sqrt{T}\right)\right| \\
		& = \frac{1}{T}\left|\Cov\left(Y,X\right)\right| \\
		& \leq \frac{1}{T}\int_0^{\gamma(\mathcal{M},z^\top\xi_r)}Q_{Y}\circ G_{z^\top\xi_{r}}(u)\dx u \\
		& \leq \frac{1}{T}\int_0^{\gamma(\mathcal{M},z^\top\xi_r)}Q_{z^\top\xi_0}\circ G_{z^\top\xi_{r}}(u)\dx u \\
		& \leq \frac{1}{T}\|\E(z^\top\xi_r|\mathcal{M}) - \E(z^\top\xi_r)\|_1^{\frac{\varsigma-2}{\varsigma-1}}\|z^\top\xi_0\|_\varsigma^{\varsigma/(\varsigma-1)} \\
		& \leq \frac{1}{T}|\Theta_G^\top z|_1^\frac{\varsigma-2}{\varsigma-1}\tau_{r}^{\frac{\varsigma-2}{\varsigma-1}}\|z^\top\xi_0\|_\varsigma^{\varsigma/(\varsigma-1)} \lesssim r^{-a\frac{\varsigma-2}{\varsigma-1}}
		\end{aligned}
	\end{equation*}
	where the second line follows by stationarity and $\sup_x|g(x)|\leq 1$, the fourth by H\"{o}lder's inequality and the change of variables
	\begin{equation*}
		\begin{aligned}
			\int Q_{z^\top\xi_{0}}^{\varsigma-1}\circ G_{z^\top\xi_{r}}(u)\dx u = \int_0^1Q_{z^\top\xi_0}^{\varsigma}(u)\dx u = \|z^\top\xi_0\|^\varsigma_\varsigma,
		\end{aligned}
	\end{equation*}
	and the last by Lemma~\ref{lemma:covariances} and Assumptions~\ref{as:data} (ii) and \ref{as:clt} (ii). Similarly,
	\begin{equation*}
		\begin{aligned}
			& \left|\Cov\left(g(z^\top\xi_{t_1}/\sqrt{T},\dots,z^\top\xi_{t_h}/\sqrt{T}),z^\top\xi_{t_h+r}/\sqrt{T}z^\top\xi_{t_h+s}/\sqrt{T}\right)\right| \\
			& = \frac{1}{T}\left|\Cov\left(g(z^\top\xi_{t_1-t_h}/\sqrt{T},\dots,z^\top\xi_{0}/\sqrt{T}), z^\top\xi_{r}z^\top\xi_{s}\right)\right| \\
			& \leq \frac{1}{T}\int_0^{\gamma(\mathcal{M},z^\top\xi_{r}z^\top\xi_{s})}Q_{g}\circ G_{z^\top\xi_{r}z^\top\xi_{s}}(u)\dx u \\
			& \leq \frac{1}{T}\left\|\E(z^\top\xi_{r}z^\top\xi_{s}|\mathcal{M}) - \E(z^\top\xi_{r}z^\top\xi_{s})\right\|_1 \\
			& \leq \frac{1}{T}|\Theta_G^\top z|_1^2\check\tau_{r} \lesssim r^{-d}.
		\end{aligned}
	\end{equation*}
	Since the sequence $(r^{-a(\varsigma-2)/(\varsigma-1)\wedge d})_{r\in\Nn}$ is summable under Assumption~\ref{as:data} (ii), all conditions of \cite{neumann2013central}, Theorem 2.1, are verified, whence $z^\top I_T \xrightarrow{d} z^\top N(0,\Xi_G)$ for every $z\in\R^{|G|}$.
	
	Next,
	\begin{equation*}
		\begin{aligned}
			|II_T|_\infty & = \left|(\hat\Theta - \Theta)_G\left(\frac{1}{\sqrt{T}}\sum_{t=1}^Tu_tx_t\right)\right|_\infty \\
			& \leq \|\hat\Theta_G - \Theta_G\|_\infty\left|\frac{1}{\sqrt{T}}\sum_{t=1}^Tu_tx_t\right|_\infty \\
			& = O_P\left(\frac{S p^{1/\kappa}}{T^{1-1/\kappa}} \vee S\sqrt{\frac{\log p}{T}}\right)O_P\left(\frac{p^{1/\kappa}}{T^{1/2-1/\kappa}} + \sqrt{\log p}\right) \\
			& = o_P(1),
		\end{aligned}
	\end{equation*}
	where the second line follows by $|Ax|_\infty \leq \|A\|_\infty|x|_\infty$, the third line by Proposition~\ref{prop:error_bounds} and the inequality in Eq.~\ref{eq:concentration} under Assumption~\ref{as:data}, and the last under Assumption~\ref{as:clt} (iv). 	Likewise, using $|Ax|_\infty \leq \max_{j,k}|A_{j,k}|_\infty|x|_1$, by Proposition~\ref{prop:error_bounds} and Theorem~\ref{thm:rates}
	\begin{equation*}
		\begin{aligned}
			|III_T|_\infty & = \sqrt{T}|(I - \hat\Theta\hat\Sigma)_G(\hat\beta - \beta)|_\infty\\
			& \leq \sqrt{T}\max_{j\in G}|(I - \hat\Theta\hat\Sigma)_j|_\infty|\hat\beta - \beta|_1 \\
			& = O_P\left(\frac{p^{1/\kappa}}{T^{1/2-1/\kappa}} \vee \sqrt{\log p}\right) O_P\left(\frac{s_\alpha p^{1/\kappa}}{T^{1-1/\kappa}} \vee s_\alpha\sqrt{\frac{\log p}{T}}\right) \\
			& = o_P(1)
		\end{aligned}
	\end{equation*}
	under Assumption~\ref{as:clt} (iv). Lastly, by the Cauchy-Schwartz inequality, under Assumption~\ref{as:clt} (v)
	\begin{equation*}
		\begin{aligned}
			|IV_T|_\infty & \leq \max_{j\in G}|\mathbf{X}\hat\Theta_j^\top|_2\|\mathbf{m}-\mathbf{X}\beta\|_T \\
			& = \max_{j\in G}\sqrt{\hat\Theta_j\hat\Sigma\hat\Theta_j^\top}o_P(1) \\
			& = o_P(1),
		\end{aligned}
	\end{equation*}
	where the last line follows since $\hat\Theta_j$ are consistent for $\Theta_j$ in the $\ell_1$ norm while $\hat\Sigma$ is consistent for $\Sigma$ in the entrywise maximum norm under the maintained assumptions.
\end{proof}

Next, we focus on the HAC estimator based on LASSO residuals. Note that by construction of the precision matrix $\hat\Theta$, its $j^{th}$ row is $\hat\Theta_jx_t = \hat v_{t,j}/\hat\sigma^2_j$, where $\hat v_{t,j}$ is the regression residual from the $j^{th}$ nodewise LASSO regression and $\hat\sigma^2_j$ is the corresponding estimator of the variance of the regression error. Therefore, the HAC estimator based on the LASSO residuals in Eq.~\ref{eq:hacformula} can be written as
\begin{equation*}
\hat\Xi_G = \sum_{|k|<T}K\left(\frac{k}{M_T}\right)\hat \Gamma_{k},
\end{equation*}
where $\hat\Gamma_{k}$ has generic $(j,h)$-entry $\frac{1}{T}\sum_{t=1}^{T-k}\hat u_t\hat u_{t+k}\hat v_{t,j}\hat v_{t+k,h}\hat \sigma_j^{-2}\hat \sigma_h^{-2}.$ 

Similarly, we define
\begin{equation*}
\tilde \Xi_G = \sum_{|k|<T}K\left(\frac{k}{M_T}\right)\tilde\Gamma_{k},
\end{equation*}
where $\tilde\Gamma_{k}$ has generic $(j,h)$-entry $\frac{1}{T}\sum_{t=1}^{T-k}u_tu_{t+k}v_{t,j}v_{t+k,h}\sigma_j^{-2}\sigma_h^{-2}$ and note that the long-run variance $\Xi_G$ has generic $(j,h)$-entry $\E[u_tu_{t+k}v_{t,j}v_{t+k,h}]\sigma_j^{-2}\sigma_h^{-2}$.

\begin{assumption}\label{as:hac_moments}
	Suppose that uniformly over $k\in\Z$ and $j,h\in G$ (i) $\E|u_0u_{k}v_{0,j}v_{k,h}|<\infty$; (ii) $\E|v_{0,j}u_{k}v_{k,h}|^2<\infty$, $\E|u_0u_{k}v_{k,h}|^2<\infty$, $\E|u_0v_{0,j}u_{k}|^2<\infty$, and $\E|u_0v_{0,j}v_{k,h}|^2<\infty$; (iii) $\E|u_0|^{2q}<\infty$ and $\E|v_{0,j}|^{2q}<\infty$ for some $q\geq 1$.
\end{assumption}

\begin{proof}[Proof of Theorem~\ref{prop:hac_residuals}]
	By Proposition~\ref{prop:HAC_true} with $V_t = (u_tv_{t,j}/\sigma^2_j)_{j\in G}$
	\begin{equation}\label{eq:hac_residuals}
	\|\hat\Xi_G - \Xi_G\| \leq \|\hat\Xi_G - \tilde\Xi_G\| + O_P\left(\sqrt{\frac{M_T}{{T}}} + M_T^{-\varsigma} + T^{-(\varsigma\wedge 1)}\right).
	\end{equation}
	Next,
	{\footnotesize
		\begin{equation*}
		\begin{aligned}
		\|\hat\Xi_G - \tilde\Xi_G\| & \leq \sum_{|k|<T}\left|K\left(\frac{k}{M_T}\right)\right|\|\hat{\Gamma}_{k} - \tilde\Gamma_{k}\| \\
		& \leq |G|\sum_{|k|<T}\left|K\left(\frac{k}{M_T}\right)\right|\max_{j,h\in G}\left|\frac{1}{\hat\sigma_j^{2}\hat\sigma_h^{2}T}\sum_{t=1}^{T-k}\hat u_t\hat u_{t+k}\hat v_{t,j}\hat v_{t+k,h} - \frac{1}{\sigma_j^{2}\sigma_h^{2}T}\sum_{t=1}^{T-k}u_tu_{t+k}v_{t,j}v_{t+k,h}\right| \\ 
		& \leq |G|\sum_{|k|<T}\left|K\left(\frac{k}{M_T}\right)\right|\max_{j,h\in G}\frac{1}{\hat\sigma_j^2\hat\sigma_h^2}\left|\frac{1}{T}\sum_{t=1}^{T-k}\hat u_t\hat u_{t+k}\hat v_{t,j}\hat v_{t+k,h} - \frac{1}{T}\sum_{t=1}^{T-k}u_tu_{t+k}v_{t,j}v_{t+k,h}\right| \\
		& \qquad + |G|\max_{j,h\in G}\left|\frac{1}{\hat\sigma_j^2\hat\sigma_h^2} - \frac{1}{\sigma_j^2\sigma_h^2}\right|\sum_{|k|<T}\left|K\left(\frac{k}{M_T}\right)\right|\left|\frac{1}{T}\sum_{t=1}^{T-k}u_tu_{t+k}v_{t,j}v_{t+k,h}\right| \\
		& \triangleq S_T^a + S_T^b.
		\end{aligned}
		\end{equation*}}
	\noindent By Proposition~\ref{thm:error_variance}, since $s_\alpha^2\log p/T\to 0$ and $s_\alpha^\kappa p/T^{4\kappa/5-1}\to 0$, under stated assumptions, we obtain $\max_{j\in G}|\hat\sigma_j^2 - \sigma_j^2| = o_P(1)$, and whence $\max_{j\in G}\hat\sigma_j^{-2} = O_P(1)$.
	Using $\hat a\hat b - ab = (\hat a - a)b + a(\hat b - b) + (\hat a - a)(\hat b - b)$, by Proposition~\ref{thm:error_variance}
	\begin{equation*}
	\begin{aligned}
	S_T^b & = O_P\left(\frac{s_\alpha p^{1/\kappa}}{T^{1-1/\kappa}} \vee s_\alpha\sqrt{\frac{\log p}{T}}\right)\sum_{|k|<T}\left|K\left(\frac{k}{M_T}\right)\right|\max_{j,h\in G}\left|\frac{1}{T}\sum_{t=1}^{T-k}u_tu_{t+k}v_{t,j}v_{t+k,h}\right| \\
	\end{aligned}
	\end{equation*}
	Under Assumptions~\ref{as:hac_moments} and (i)~\ref{as:HAC} (i)
	{\footnotesize \begin{equation*}
		\begin{aligned}
		\E\left[\sum_{|k|<T}\left|K\left(\frac{k}{M_T}\right)\right|\max_{j,h\in G}\left|\frac{1}{T}\sum_{t=1}^{T-k}u_tu_{t+k}v_{t,j}v_{t+k,h}\right|\right] & \leq O(M_T)\sup_{k\in \Z}\sum_{j,h\in G}\E\left|\frac{1}{T}\sum_{t=1}^{T-k}u_tu_{t+k}v_{t,j}v_{t+k,h}\right| \\
		& \leq O(M_T)|G|^2\sup_{k\in \Z}\max_{j,h\in G}\E|u_tu_{t+k}v_{t,j}v_{t+k,h}| \\ &  = O(M_T),
		\end{aligned}
		\end{equation*}}
	and whence $S_T^b = O_P\left(M_T\left(\frac{s_\alpha p^{1/\kappa}}{T^{1-1/\kappa}} \vee s_\alpha\sqrt{\frac{\log p}{T}}\right)\right)$.
	
	Next, we evaluate uniformly over $|k|<T$
	\begin{equation*}
	\begin{aligned}
	& \left|\frac{1}{T}\sum_{t=1}^{T-k}\hat u_t\hat u_{t+k}\hat v_{t,j}\hat v_{t+k,h} - \frac{1}{T}\sum_{t=1}^{T-k}u_tu_{t+k}v_{t,j}v_{t+k,h}\right| \\
	& \leq \left|\frac{1}{T}\sum_{t=1}^{T-k}(\hat u_t\hat v_{t,j} - u_tv_{t,j})u_{t+k}v_{t+k,h}\right| + \left|\frac{1}{T}\sum_{t=1}^{T-k}u_tv_{t,j}(\hat u_{t+k}\hat v_{t+k,h} - u_{t+k}v_{t+k,h})\right| \\
	& \quad + \left|\frac{1}{T}\sum_{t=1}^{T-k}(\hat u_t\hat v_{t,j} - u_tv_{t,j})(\hat u_{t+k}\hat v_{t+k,h} - u_{t+k}v_{t+k,h})\right|  \triangleq I_T + II_T + III_T.
	\end{aligned}
	\end{equation*}
	We bound the first term as
	\begin{equation*}
	\begin{aligned}
	I_T & \leq \left|\frac{1}{T}\sum_{t=1}^{T-k}(\hat u_t - u_t)v_{t,j}u_{t+k}v_{t+k,h}\right| + \left|\frac{1}{T}\sum_{t=1}^{T-k}u_t(\hat v_{t,j} - v_{t,j})u_{t+k}v_{t+k,h}\right| \\
	& \quad + \left|\frac{1}{T}\sum_{t=1}^{T-k}(\hat u_t - u_t)(\hat v_{t,j} - v_{t,j})u_{t+k}v_{t+k,h}\right| \triangleq I_T^a + I_T^b + I_T^c.
	\end{aligned}
	\end{equation*}
	By the Cauchy-Schwartz inequality, under Assumptions of Theorem~\ref{thm:rates} for the sg-LASSO and Assumption~\ref{as:hac_moments} (ii)
	\begin{equation*}
	\begin{aligned}
	I_T^a & = \left|\frac{1}{T}\sum_{t=1}^{T-k}\left(x_t^\top(\beta - \hat\beta) + m_t - x_t^\top\beta\right)v_{t,j}u_{t+k}v_{t+k,h}\right| \\
	& \leq ( \|\mathbf{X}(\hat\beta - \beta)\|_T + \|\mathbf{m}-\mathbf{X}\beta\|_T)\sqrt{\frac{1}{T}\sum_{t=1}^{T-k}v_{t,j}^2u_{t+k}^2v_{t+k,h}^2} \\
	& = O_P\left(\frac{s_\alpha p^{1/\kappa}}{T^{1-1/\kappa}} \vee \sqrt{\frac{s_\alpha\log p}{T}}\right).
	\end{aligned}
	\end{equation*}
	Similarly, under Assumptions of Theorem~\ref{thm:rates} for the nodewise LASSO and Assumption~\ref{as:hac_moments} (ii)
	\begin{equation*}
	\begin{aligned}
	I_T^b & \leq \left(\|\mathbf{X}_{-j}(\hat\gamma_j - \gamma_j)\|_T + o_P(T^{-1/2})\right)\sqrt{\frac{1}{T}\sum_{t=1}^{T-k}u_{t}^2u_{t+k}^2v_{t+k,h}^2} = O_P\left(\frac{S_j p^{1/\kappa}}{T^{1-1/\kappa}} \vee \sqrt{\frac{S_j\log p}{T}}\right).
	\end{aligned}
	\end{equation*}
	Note that for arbitrary $(\xi_t)_{t\in\Z}$ and $q\geq 1$, by Jensen's inequality
	\begin{equation*}
	\begin{aligned}
	\E\left[\max_{t\in[T]}|\xi_t|\right]  \leq \left(\E\left[\max_{t\in[T]}|\xi_t|^q\right]\right)^{1/q} \leq \left(\E\left[\sum_{t=1}^T|\xi_t|^q\right]\right)^{1/q} = T^{1/q}\left(\E|\xi_t|^q\right)^{1/q}.
	\end{aligned}
	\end{equation*}
	Then by the Cauchy-Schwartz inequality under Assumption~\ref{as:hac_moments} (iii) and Theorem~\ref{thm:rates}
	\begin{equation*}
	\begin{aligned}
	I_T^c & \leq (\|\mathbf{X}(\hat\beta - \beta)\|_T + o_P(T^{-1/2}))(\|\mathbf{X}_{-j}(\hat\gamma_j - \gamma_j)\|_T + o_P(T^{-1/2})) \max_{t\in[T]}|u_tv_{t,h}| \\
	& = O_P\left(\frac{s^2p^{2/\kappa}}{T^{2-3/\kappa}} \vee \frac{s\log p}{T^{1-1/\kappa}}\right),
	\end{aligned}
	\end{equation*}
	where we use the fact that $\kappa\leq q$. Therefore, under maintained assumptions
	\begin{equation*}
	I_T = O_P\left(\frac{s p^{1/\kappa}}{T^{1-1/\kappa}} \vee \sqrt{\frac{s\log p}{T}} + \frac{s^2p^{2/\kappa}}{T^{2-3/\kappa}} \vee \frac{s\log p}{T^{1-1/\kappa}}\right)
	\end{equation*}
	and by symmetry
	\begin{equation*}
	II_T = O_P\left(\frac{s p^{1/\kappa}}{T^{1-1/\kappa}} \vee \sqrt{\frac{s\log p}{T}} + \frac{s^2p^{2/\kappa}}{T^{2-3/\kappa}} \vee \frac{s\log p}{T^{1-1/\kappa}}\right).
	\end{equation*}
	Lastly, by the Cauchy-Schwartz inequality
	\begin{equation*}
	\begin{aligned}
	III_T & \leq \sqrt{\frac{1}{T}\sum_{t=1}^{T-k}(\hat u_t\hat v_{t,j} - u_tv_{t,j})^2\frac{1}{T}\sum_{t=1}^{T-k}(\hat u_{t+k}\hat v_{t+k,h} - u_{t+k}v_{t+k,h})^2} \\
	& \leq \sqrt{\frac{1}{T}\sum_{t=1}^{T}(\hat u_t\hat v_{t,j} - u_tv_{t,j})^2\frac{1}{T}\sum_{t=1}^{T}(\hat u_{t}\hat v_{t,h} - u_{t}v_{t,h})^2}. \\
	\end{aligned}
	\end{equation*}
	For each $j\in G$
	\begin{equation*}
	\begin{aligned}	
	\frac{1}{T}\sum_{t=1}^T(\hat u_t\hat v_{t,j} - u_tv_{t,j})^2 & \leq \frac{3}{T}\sum_{t=1}^T|\hat u_t - u_t|^2v_{t,j}^2 + \frac{3}{T}\sum_{t=1}^T|\hat v_{t,j} - v_{t,j}|^2u_t^2 \\
	& \qquad + \frac{3}{T}\sum_{t=1}^T|\hat u_t - u_t|^2|\hat v_{t,j} - v_{t,j}|^2 \\
	& \triangleq III_T^a + III_T^b + III_T^c.
	\end{aligned}
	\end{equation*}
	Since under Assumption~\ref{as:hac_moments} (iii), $\E|v_{t,j}|^{2q}<\infty$ and $\E|u_t|^{2q}<\infty$,
	\begin{equation*}
	\begin{aligned}
	III_T^a & \leq 3\max_{t\in[T]}|v_{t,j}|^2(\|\mathbf{X}(\hat\beta - \beta)\|_T^2 + o_P(T^{-1/2}))= O_P\left(\frac{s_\alpha p^{2/\kappa}}{T^{2-3/\kappa}} \vee \frac{s_\alpha\log p}{T^{1-1/\kappa}}\right)
	\end{aligned}
	\end{equation*}
	and
	\begin{equation*}
	\begin{aligned}
	III_T^b & \leq 3\max_{t\in[T]}|u_t|^2(\|\mathbf{X}_{-j}(\hat\gamma_j - \gamma_j)\|^2_T + o_P(T^{-1/2})) = O_P\left(\frac{S_j p^{2/\kappa}}{T^{2-3/\kappa}} \vee \frac{S_j\log p}{T^{1-1/\kappa}}\right).
	\end{aligned}
	\end{equation*}
	For the last term, since under Assumption~\ref{as:data} (ii), $\sup_{k}\E|X_{t,k}|^{2\tilde q}<\infty$ and $\kappa\geq \tilde q$, by Theorem~\ref{thm:rates}
	\begin{equation*}
	\begin{aligned}
	III_T^c & \leq 3(\|\mathbf{X}(\hat\beta - \beta)\|_T^2 + o_P(T^{-1/2})) \max_{t\in[T]}|X_{t,-j}^\top(\hat\gamma_j - \gamma_j) + m_t - X_t^\top\beta|^2 \\
	& \leq O_P\left(\frac{s_\alpha p^{2/\kappa}}{T^{2-2/\kappa}}\vee \frac{s_\alpha\log p}{T}\right)\left(2\max_{t\in[T]}|X_{t}|_\infty^2|\hat\gamma_j - \gamma_j|_1^2 + 2T\|\mathbf{m} - \mathbf{X}^\top\beta\|_T^2\right) \\
	& = O_P\left(\left(\frac{s_\alpha p^{2/\kappa}}{T^{2-2/\kappa}}\vee \frac{s_\alpha\log p}{T}\right)\left(\frac{S^2p^{2/\kappa}}{T^{2-2/\kappa}} \vee S^2\frac{\log p}{T} \right)(pT)^{1/\kappa}\right) \\
	& = O_P\left(\frac{s^3p^{5/\kappa}}{T^{4-5/\kappa}} + \frac{s^3p^{3/\kappa}\log p}{T^{3-3/\kappa}} + \frac{s^3p^{1/\kappa}\log^2 p}{T^{2-1/\kappa}}\right) \\
	& = O_P\left(\frac{s^3p^{5/\kappa}}{T^{4-5/\kappa}} + \frac{s^3p^{3/\kappa}\log p}{T^{3-3/\kappa}}\right),
	\end{aligned}
	\end{equation*}
	where we use the fact that $\kappa>2$, $s = s_\alpha\vee S$, $s^\kappa p=o(T^{4\kappa/5-1})$, and $s^2\log p/T\to 0$ as $T\to\infty$.
	Then for every $j\in G$
	\begin{equation*}
	\frac{1}{T}\sum_{t=1}^T(\hat u_t\hat v_{t,j} - u_tv_{t,j})^2 = O_P\left(\frac{s p^{2/\kappa}}{T^{2-3/\kappa}} \vee \frac{s\log p}{T^{1-1/\kappa}} + \frac{s^3p^{5/\kappa}}{T^{4-5/\kappa}} + \frac{s^3p^{3/\kappa}\log p}{T^{3-3/\kappa}}\right),
	\end{equation*}
	and whence
	\begin{equation*}
	III_T = O_P\left(\frac{s p^{2/\kappa}}{T^{2-3/\kappa}} \vee \frac{s\log p}{T^{1-1/\kappa}} + \frac{s^3p^{5/\kappa}}{T^{4-5/\kappa}}+ \frac{s^3p^{3/\kappa}\log p}{T^{3-3/\kappa}}\right).
	\end{equation*}
	Therefore, since $\hat\sigma_j^2 \xrightarrow{P}\sigma_j^2$, we obtain
	\begin{equation*}
	\begin{aligned}
	S_T^a & = O_P\left(M_T\left(\frac{s p^{1/\kappa}}{T^{1-1/\kappa}} \vee \sqrt{\frac{s\log p}{T}} + \frac{s^2p^{2/\kappa}}{T^{2-3/\kappa}}\vee \frac{s\log p}{T^{1-1/\kappa}} + \frac{s^3p^{5/\kappa}}{T^{4-5/\kappa}} + \frac{s^3p^{3/\kappa}\log p}{T^{3-3/\kappa}}\right)\right) \\
	& = O_P\left(M_T\left(\frac{s p^{1/\kappa}}{T^{1-1/\kappa}} \vee \sqrt{\frac{s\log p}{T}} + \frac{s^2p^{2/\kappa}}{T^{2-3/\kappa}} + \frac{s^3p^{5/\kappa}}{T^{4-5/\kappa}}\right)\right), \\
	\end{aligned}
	\end{equation*}
	where the last line follows since $s^\kappa p/T^{4\kappa/5-1}=o(1)$. Combining this estimate with previously obtained estimate for $S_T^b$
	\begin{equation*}
	\|\hat\Xi_G - \tilde\Xi_G\| = O_P\left(M_T\left(\frac{s p^{1/\kappa}}{T^{1-1/\kappa}} \vee s\sqrt{\frac{\log p}{T}} + \frac{s^2p^{2/\kappa}}{T^{2-3/\kappa}} + \frac{s^3p^{5/\kappa}}{T^{4-5/\kappa}}\right)\right).
	\end{equation*}
	The result follows from combining this estimate with the estimate in equation (\ref{eq:hac_residuals}).
\end{proof}

\begin{proof}[Proof of Theorem~\ref{thm:tails_polynomial}]
	Suppose first that $p=1$. For $a\in\R$, with some abuse of notation, let $[a]$ denote its integer part. We split partial sums into blocks $V_k = \xi_{(k-1)J + 1} + \dots + \xi_{kJ},k=1,2,\dots,[T/J]$ and $V_{[T/J]+1} = \xi_{[T/J]J+1} + \dots + \xi_T$, where we set $V_{[T/J]+1}=0$ if $T/J$ is an integer. Let $\{U_t:\;t=1,2,\dots, [T/J]+1\}$ be i.i.d.\ random variables drawn from the uniform distribution on $(0,1)$ independently of $\{V_t:\;t=1,2,\dots, [T/J]+1\}$. Put $\mathcal{M}_t = \sigma(V_1,\dots,V_{t-2})$ for every $t=3,\dots,[T/J]+1$. Next, for $t=1,2$, set $V_t^* = V_t$, while for $t\geq 3$, by \cite{dedecker2004coupling}, Lemma 5, there exist random variables $V_t^*=_dV_t$ such that:
	\begin{enumerate}
		\item $V_t^*$ is $\sigma(V_1,\dots, V_{t-2})\vee\sigma(V_t)\vee\sigma(U_t)$-measurable;
		\item $V_t^*\si(V_1,\dots,V_{t-2})$;
		\item $\|V_t - V_t^*\|_1 = \tau(\mathcal{M}_t,V_t)$.
	\end{enumerate}
	It follows from properties 1. and 2. that $(V_{2t}^*)_{t\geq 1}$ and $(V_{2t-1}^*)_{t\geq 1}$ are sequences of independent random variables. Then
	\begin{equation*}
	\begin{aligned}
	\left|\sum_{t=1}^T\xi_t\right| & \leq \left|\sum_{t\geq 1}V_{2t}^*\right| + \left|\sum_{t\geq 1}V_{2t-1}^*\right| + \sum_{t=3}^{[T/J]+1}|V_t - V_t^*| \\
	& \triangleq I_T + II_T + III_T.
	\end{aligned}
	\end{equation*}
	
	By \cite{fuk1971probability}, Corollary 4, there exist constants $c_q^{(j)},j=1,2$ such that
	\begin{equation*}
	\begin{aligned}
	\Pr(I_T \geq x) & \leq \frac{c_q^{(1)}}{x^q}\sum_{t\geq 1}\E|V_{2t}^*|^q + 2\exp\left(-\frac{c_q^{(2)}x^2}{\sum_{t\geq 1}\Var(V_{2t}^*)}\right) \\
	& \leq \frac{c_q^{(1)}}{x^q}\sum_{t\geq 1}\E|V_{2t}|^q + 2\exp\left(-\frac{c_q^{(2)}x^2}{B^2_T}\right),
	\end{aligned}
	\end{equation*}
	where the second inequality follows since $\sum_{t\geq 1}\Var(V_{2t}^*) = \sum_{t\geq 1}\Var(V_{2t})\leq B_T^2$. Similarly
	\begin{equation*}
	\Pr(II_T \geq x) \leq \frac{c_q^{(1)}}{x^q}\sum_{t\geq 1}\E|V_{2t-1}|^q + 2\exp\left(-\frac{c_q^{(2)}x^2}{B^2_T}\right). \\
	\end{equation*}	
	Lastly, by Markov's inequality and property 3.
	\begin{equation*}
	\begin{aligned}
	\Pr\left(III_T \geq x\right) & \leq \frac{1}{x}\sum_{t=3}^{[T/J]+1} \tau(\mathcal{M}_t,V_t) \\
	& \leq \frac{1}{x}\sum_{t=3}^{[T/J]+1} \tau(\mathcal{M}_t,(\xi_{(t-1)J+1},\dots,\xi_{tJ})) \\
	& \leq \frac{1}{x}[T/J]\sup_{t+J+1\leq t_1<\dots<t_J}\tau(\mathcal{M}_t,(\xi_{t_1},\dots,\xi_{t_J})) \\
	& \leq \frac{T}{x}\tau_{J+1},
	\end{aligned}
	\end{equation*}
	where the second inequality follows since the sum is a $1$-Lipschitz function with respect to $|.|_1$-norm and the third since $\mathcal{M}_t$ and $(\xi_{(t-1)J+1},\dots,\xi_{tJ})$ are separated by $J+1$ lags of $(\xi_t)_{t\in\Z}$.
	
	Combining all the estimates together
	\begin{equation*}
	\begin{aligned}
	\Pr\left(\left|\sum_{t=1}^T\xi_t\right|\geq 3x\right) & \leq \Pr(I_T\geq x) + \Pr(II_T \geq x) + \Pr(III_T \geq x) \\
	& \leq \frac{c_q^{(1)}}{x^q}\sum_{t=1}^{[T/J]+1}\E|V_{t}|^q + 4\exp\left(-\frac{c_q^{(2)}x^2}{B^2_T}\right) + \frac{T}{x}\tau_{J+1} \\
	& \leq \frac{c_q^{(1)}}{x^q}J^{q-1}\sum_{t=1}^T\|\xi_t\|_q^q + \frac{T}{x}c(J+1)^{-a} + 4\exp\left(-\frac{c_q^{(2)}x^2}{B^2_T}\right).
	\end{aligned}
	\end{equation*}	
	To balance the first two terms, we shall set $J\sim x^\frac{q-1}{q+a-1}$, in which case we obtain the result under maintained assumptions. The result for $p>1$ follows by the union bound.
\end{proof}

For a stationary process $(\xi_t)_{t\in\Z}$, let
\begin{equation*}
	\gamma_k = \|\E(\xi_k|\mathcal{M}_0) - \E(\xi_k)\|_1
\end{equation*}
be its $L_1$ mixingale coefficient with respect to the canonical filtration $\mathcal{M}_0=\sigma(\xi_0,\xi_{-1},\xi_{-2},\dots)$. Let $\alpha_k$  be the $\alpha$-mixing coefficient and let $Q$ be the quantile function of $|\xi_0|$. The following covariance inequality allows us controlling the autocovariances in terms of the $\tau$-mixing coefficient as well as comparing the latter to the mixingale and the $\alpha$-mixing coefficients. 

\begin{lemma}\label{lemma:covariances}
	Let $(\xi_t)_{t\in\Z}$ be a centered stationary stochastic process with $\|\xi_0\|_q<\infty$ for some $q>2$. Then
	\begin{equation*}
	|\Cov(\xi_0,\xi_t)| \leq \gamma_t^{\frac{q-2}{q-1}}\|\xi_0\|_q^{q/(q-1)}
	\end{equation*}
	and 
	\begin{equation*}
	\gamma_t \leq \tau_t\leq 2\int_0^{2\alpha_k}Q(u)\dx u.
	\end{equation*}
\end{lemma}
\begin{proof}
	Let $G$ be the generalized inverse of $x\mapsto\int_0^xQ(u)\dx u$. By \cite{dedecker2003new}, Proposition 1
	\begin{equation*}
	\begin{aligned}
	|\Cov(\xi_0,\xi_t)| & \leq \int_0^{\gamma_t}(Q\circ G)(u)\dx u \\
	& \leq \gamma_t^{\frac{q-2}{q-1}}\left(\int_0^{\|\xi_0\|_1}(Q\circ G)^{q-1}(u)\dx u\right)^{1/(q-1)} \\
	& = \gamma_t^{\frac{q-2}{q-1}}\|\xi_0\|_q^{q/(q-1)},
	\end{aligned}
	\end{equation*}
	where the second line follows by H\"{o}lder's inequality and the last equality by the change of variables
	$\int_0^{\|\xi_0\|_1}(Q\circ G)^{q-1}(u)\dx u$ = $\int_0^1Q^q(u)\dx u$ = $\E|\xi_0|^q.$
	The second statement follows from \cite{dedecker2003new}, Lemma 1 and \cite{dedecker2004coupling}, Lemma 6.
\end{proof}

The following result shows that the variance of partial sums can be controlled provided that the $\tau$-mixing coefficients decline sufficiently fast.
\begin{lemma}\label{lemma:B_T_bound}
	Let $(\xi_t)_{t\in\Z}$ be a centered stationary stochastic process such that $\|\xi_t\|_q<\infty$ for some $q>2$ and $\tau_k=O(k^{-a})$ for some $a>\frac{q-1}{q-2}$. Then
	\begin{equation*}
	\sum_{t=1}^T\sum_{k=1}^T|\Cov(\xi_{t,j},\xi_{k,j})| = O(T).
	\end{equation*}
\end{lemma}
\begin{proof}
	Under stationarity
	\begin{equation*}
	\begin{aligned}
	\sum_{t=1}^T\sum_{k=1}^T|\Cov(\xi_{t,j},\xi_{k,j})| & = T\Var(\xi_0) + 2\sum_{k=1}^{T-1}\left(T-k\right)\Cov(\xi_{0},\xi_k) \\
	& \leq T\Var(\xi_0) + 2T\|\xi_{t}\|_q^{q/(q-1)}\sum_{k=1}^{T-1}\tau_k^{\frac{q-2}{q-1}}\\
	& = O(T),
	\end{aligned}
	\end{equation*}	
	where the second line follows by Proposition~\ref{lemma:covariances} and the last since the series $\sum_{k=1}^\infty k^{-a\frac{q-2}{q-1}}$ converges under the maintained assumptions.
\end{proof}

Lastly, we show that in the linear regression setting, the long-run variance for a group of projection coefficients $G\subset[p]$ of fixed size exists under mild conditions. Let $\xi_t=u_t\Theta_Gx_t$, where $\Theta_G$ are rows of the precision matrix $\Theta=\Sigma^{-1}$ corresponding to indices in $G$.
\begin{proposition}\label{prop:long_run}
	Suppose that (i) $(u_tx_t)_{t\in\Z}$ is stationary for every $p\geq 1$; (ii) $\|u_t\|_q<\infty$ and $\max_{j\in[p]}\|x_{t,j}\|_r=O(1)$ for some $q>2r/(r-2)$ and $r>4$; (iii) for every $j\in[p]$, the $\tau$-mixing coefficients of $(u_tx_{t,j})_{t\in\Z}$ are $\tau_k\leq ck^{-a}$ for all $k\geq 0$, where $c>0$ and $a>(\varsigma-1)/(\varsigma-2)$, and $\varsigma = qr/(q+r)$ are some universal constants; (iv) $\|\Theta_G\|_\infty = O(1)$ and $\sup_x\E[|u_t|^2|x_t=x]=O(1)$. Then for every $z\in\R^{|G|}$, the limit
	\begin{equation*}
		\lim_{T\to\infty}\Var\left(\frac{1}{\sqrt{T}}\sum_{t=1}^Tz^\top\xi_t\right)
	\end{equation*}
	exists.
\end{proposition}
\begin{proof}
	Under assumption (i), by H\"{o}lder's inequality, $\max_{j\in[p]}\|u_tx_{t,j}\|_\varsigma = O(1)$ with $\varsigma=qr/(q+r)$, whence
	by the Minkowski inequality and assumption (iv)
	\begin{equation}\label{eq:moment_q}
	\begin{aligned}
	\|z^\top\xi_t\|_{\varsigma} & \leq  \sum_{k\in G}\sum_{j\in[p]}|\Theta_{k,j}|\|u_tx_{t,j}\|_\varsigma \leq |G|\|\Theta_G\|_\infty = O(1).
	\end{aligned}
	\end{equation}
	Since $\varsigma>2$, this shows that $\Var(z^\top\xi_0)$ exists. Moreover,
	\begin{equation*}
	\begin{aligned}
		\Var(z^\top\xi_0) & = z^\top\Theta_G\Var(u_0x_0)\Theta_G^\top z \\
		& = \sum_{j,k\in[p]}(z^\top\Theta_G)_j(z^\top\Theta_G)_k\E[u_0^2x_{0,j}x_{0,k}],
	\end{aligned}
	\end{equation*}
	where the sum converges as $p\to\infty$ by the comparison test under assumption (iv) implying that $\lim_{T\to\infty}\Var(z^\top\xi_0)$ exists. Next, under assumption (i), for every $z\in\R^{|G|}$
	\begin{equation*}
	\Var\left(\frac{1}{\sqrt{T}}\sum_{t=1}^Tz^\top\xi_t\right) = \Var(z^\top\xi_0) + 2\sum_{k=1}^{T-1}\left(1-\frac{k}{T}\right)\Cov(z^\top\xi_0, z^\top\xi_k).
	\end{equation*}
	By Lemma~\ref{lemma:covariances}, we bound covariances by the mixingale coefficient for every $k\geq 1$
	\begin{equation*}
	\begin{aligned}
		|\Cov(z^\top\xi_0, z^\top\xi_k)| & \lesssim \|z^\top\xi_0\|_\varsigma^{\varsigma/(\varsigma-1)}\|\E(z^\top\xi_k|\mathcal{M}_0) - \E(z^\top\xi_k)\|_1^{\frac{\varsigma-2}{\varsigma-1}} \\
		& \lesssim |z^\top\Theta_G|_1^\frac{\varsigma-2}{\varsigma-1}\max_{j\in[p]}\|\E(u_kx_{k,j}|\mathcal{M}_0) - \E(u_kx_{k,j})\|_1^{\frac{\varsigma-2}{\varsigma-1}}\\
		& \lesssim \tau_k^\frac{\varsigma-2}{\varsigma-1}
	\end{aligned}
	\end{equation*}
	where the first inequality follows by Lemma~\ref{lemma:covariances}, the second by equation~(\ref{eq:moment_q}), and the third by Lemma~\ref{lemma:covariances}. Under assumption (iii), $\sum_{k=1}^\infty\tau_k^{(\varsigma-2)/(\varsigma-1)}$ converges, which implies that
	\begin{equation*}
		\sum_{k=1}^\infty|\Cov(z^\top\xi_0, z^\top\xi_k)|<\infty
	\end{equation*}
	by the comparison test. Therefore, by Lebesgue's dominated convergence, this shows that the long run variance
	\begin{equation*}
	\lim_{T\to\infty}\Var\left(\frac{1}{\sqrt{T}}\sum_{t=1}^Tz^\top\xi_t\right)
	\end{equation*}	
	exists.
\end{proof}

\newpage

\bigskip
\begin{center}
	{\large\bf SUPPLEMENTARY MATERIAL}
\end{center}

\begin{description}
	
	\item[Additional results and proofs:] This file contains supplementary results with proofs. \label{supplementary_material}
	
\end{description}
We recall first the convergence rates for the sg-LASSO with weakly dependent data that will be needed throughout the paper from \cite{babii2020midasml}, Corollary 3.1.
\begin{theorem}\label{thm:rates}
	Suppose that Assumptions~\ref{as:data}, \ref{as:covariance}, \ref{as:tuning}, and \ref{as:rates} are satisfied. Then
	\begin{equation*}
	\|\mathbf{X}(\hat\beta - \beta)\|_T^2 = O_P\left(\frac{s_\alpha p^{2/\kappa}}{T^{2-2/\kappa}}\vee\frac{s_\alpha\log p}{T}\right).
	\end{equation*}
	and
	\begin{equation*}
	\Omega(\hat{\beta} - \beta) = O_P\left(\frac{s_\alpha p^{1/\kappa}}{T^{1-1/\kappa}}\vee s_\alpha\sqrt{\frac{\log p}{T}}\right).
	\end{equation*}
\end{theorem}

Next, we consider the regularized estimator of the variance of the regression error
\begin{equation*}
\hat\sigma^2 = \|\mathbf{y} - \mathbf{X}\hat\beta\|^2_T + \lambda\Omega(\hat\beta),
\end{equation*}
where $\hat\beta$ is the sg-LASSO estimator. While the regularization is not needed to have a consistent variance estimator, the LASSO version of the regularized estimator ($\alpha=1$) is needed to establish the CLT for the debiased sg-LASSO estimator. The following result describes the converges of this variance estimator to its population counterpart $\sigma^2 = \E\|\mathbf{u}\|^2_T$.
\begin{proposition}\label{thm:error_variance}
	Suppose that Assumptions~\ref{as:data}, \ref{as:covariance}, \ref{as:tuning}, and \ref{as:rates} are satisfied and that $(u_t^2)_{t\in\Z}$ has a finite long run variance. Then
	\begin{equation*}
	\hat\sigma^2 = \sigma^2 + O_P\left(\frac{s_\alpha p^{1/\kappa}}{T^{1-1/\kappa}} \vee s_\alpha\sqrt{\frac{\log p}{T}}\right)
	\end{equation*}
	provided that $\frac{s_\alpha p^{1/\kappa}}{T^{1-1/\kappa}} \vee s_\alpha\sqrt{\frac{\log p}{T}}=o(1)$.
\end{proposition}
\begin{proof}
	We have
	\begin{equation*}
	\begin{aligned}
	|\hat\sigma^2 - \sigma^2| & = \left|\|\mathbf{u}\|_T^2 + 2\langle \mathbf{u},\mathbf{m} - \mathbf{X}\hat\beta \rangle_T - \|\mathbf{m} - \mathbf{X}\hat\beta\|_T^2 + \lambda\Omega(\hat\beta) - \sigma^2\right| \\
	& \leq |\sigma^2 - \|\mathbf{u}\|^2_T| +  2\|\mathbf{u}\|_T\|\mathbf{m} - \mathbf{X}\hat\beta\|_T + 2\|\mathbf{X}(\hat\beta - \beta)\|_T^2 + 2\|\mathbf{m} - \mathbf{X}\beta\|_T^2 + \lambda\Omega(\hat\beta) \\
	& \triangleq I_T + II_T + III_T + IV_T + V_T.
	\end{aligned}
	\end{equation*}
	By the Chebychev's inequality since the long-run variance exists, for every $\varepsilon>0$
	\begin{equation*}
	\Pr\left(\left|\frac{1}{\sqrt{T}}\sum_{t=1}^T(u_t^2 - \sigma^2)\right|>\varepsilon\right) \leq \frac{1}{\varepsilon^2}\sum_{t\in \Z}\Cov(u_0^2,u_t^2),
	\end{equation*}
	whence $I_T = O_P\left(\frac{1}{\sqrt{T}}\right)$. Therefore, by the triangle inequality and Theorem~\ref{thm:rates}
	\begin{equation*}
	\begin{aligned}
	II_T & = O_P(1)\|\mathbf{m} - \mathbf{X}\hat\beta\|_T \\
	& \leq O_P(1)\left(\|\mathbf{m} - \mathbf{X}\beta\|_T + \|\mathbf{X}(\hat\beta-\beta)\|_T\right) = O_P\left(s_\alpha^{1/2}\lambda + \frac{s_\alpha^{1/2} p^{1/\kappa}}{T^{1-1/\kappa}} \vee \sqrt{\frac{s_\alpha\log p}{T}}\right).
	\end{aligned}
	\end{equation*}
	By Theorem~\ref{thm:rates} we also have
	\begin{equation*}
	III_T + IV_T =O_P\left(\frac{s_\alpha p^{2/\kappa}}{T^{2-2/\kappa}} \vee \frac{s_\alpha\log p}{T} + s_\alpha\lambda^2\right).
	\end{equation*}
	Lastly, another application of Theorem~\ref{thm:rates} gives
	\begin{equation*}
	\begin{aligned}
	V_T & = \lambda\Omega(\hat\beta - \beta) + \lambda\Omega(\beta) \\
	& = O_P\left(\lambda\left(\frac{s_\alpha p^{1/\kappa}}{T^{1-1/\kappa}} \vee s_\alpha\sqrt{\frac{\log p}{T}}\right) + \lambda s_\alpha\right).
	\end{aligned}
	\end{equation*}
	The result follows from combining all estimates together.
\end{proof}

Next, we look at the estimator of the precision matrix. Consider nodewise LASSO regressions in equation (\ref{eq:nodewise_regressions}) for each $j\in[p]$. Put $S=\max_{j\in G}S_j$, where $S_j$ is the support of $\gamma_j$.
\begin{proposition}\label{prop:error_bounds}
	Suppose that Assumptions~\ref{as:data}, \ref{as:covariance}, \ref{as:tuning}, and \ref{as:rates} are satisfied for each nodewise regression $j\in G$ and that $(v_{t,j}^2)_{t\in\Z}$ has a finite long-run variance for each $j\in G$. Then if $S^\kappa pT^{1-\kappa}\to0$ and $S^2\log p/T\to 0$
	\begin{equation*}
	\|\hat\Theta_G - \Theta_G\|_\infty = O_P\left(\frac{S p^{1/\kappa}}{T^{1-1/\kappa}} \vee S\sqrt{\frac{\log p}{T}}\right)
	\end{equation*}
	and 
	\begin{equation*}
	\max_{j\in G}|(I - \hat\Theta\hat\Sigma)_j|_\infty = O_P\left(\frac{p^{1/\kappa}}{T^{1-1/\kappa}} \vee \sqrt{\frac{\log p}{T}}\right).
	\end{equation*}
\end{proposition}
\begin{proof}
	By Theorem~\ref{thm:rates} and Proposition~\ref{thm:error_variance} with $\alpha=1$ (corresponding to the LASSO estimator of $\gamma_j$ and $\sigma_j^2$)
	\begin{equation*}
	\begin{aligned}
	\|\hat\Theta_G - \Theta_G\|_\infty & = \max_{j\in G}|\hat\Theta_j - \Theta_j|_1 \\
	& \leq  \max_{j\in G}\left\{|\hat\gamma_j|_1\left|\hat\sigma_j^{-2} - \sigma_j^{-2}\right| + |\hat\gamma_j - \gamma_j|_1|\sigma_j^{-2}| \right\} \\
	& = O_P\left(\frac{Sp^{1/\kappa}}{T^{1-1/\kappa}} \vee S\sqrt{\frac{\log p}{T}}\right),
	\end{aligned}
	\end{equation*}
	where we use the fact that $|G|$ is fixed and that $\hat\sigma_j^2 \xrightarrow{p}\sigma_j^2$ under maintained assumptions.
	
	Second, for each $j\in G$, by Fermat's rule,
	\begin{equation*}
	\mathbf{X}_{-j}^\top(\mathbf{X}_j - \mathbf{X}_{-j}\hat\gamma_j)/T = \lambda_j z^*,\qquad z^*\in\partial|\hat\gamma_j|_1,
	\end{equation*}
	where $\hat\gamma_j^\top z^* = |\hat\gamma_j|_1$ and $|z^*|_\infty\leq 1$. Then
	\begin{equation*}
	\begin{aligned}
	\mathbf{X}_j^\top(\mathbf{X}_j - \mathbf{X}_{-j}\hat\gamma_j)/T & = \|\mathbf{X}_j - \mathbf{X}_{-j}\hat\gamma_j\|_T^2 + \hat\gamma_j^\top\mathbf{X}_{-j}^\top(\mathbf{X}_j - \mathbf{X}_{-j}\hat\gamma_j)/T \\
	& = \|\mathbf{X}_j - \mathbf{X}_{-j}\hat\gamma_j\|_T^2 + \lambda_j\hat\gamma_j^\top z^* = \hat\sigma_j^2,
	\end{aligned}
	\end{equation*}	
	and whence
	\begin{equation*}
	\begin{aligned}
	|(I - \hat\Theta\hat\Sigma)_j|_\infty & = |I_j - (\mathbf{X}_j - \mathbf{X}_{-j}\hat\gamma_j)^\top\mathbf{X}/(T\hat\sigma_j^2)|_\infty \\
	& = \max\left\{|1 - \mathbf{X}_j^\top(\mathbf{X}_j - \mathbf{X}_{-j}\hat\gamma_j)/(T\hat\sigma_j^2)|, |\mathbf{X}_{-j}^\top(\mathbf{X}_j - \mathbf{X}_{-j}\hat\gamma_j)/(T\hat\sigma_j^2)|_\infty \right\} \\
	& = \lambda_j|z^*|_\infty/\hat\sigma_j^2 = O_P\left(\frac{p^{1/\kappa}}{T^{1-1/\kappa}} \vee \sqrt{\frac{\log p}{T}}\right),
	\end{aligned}
	\end{equation*}
	where the last line follows since $\hat\sigma_j^{-2} = O_P(1)$ and $|z^*|_\infty\leq 1$. The conclusion follows from the fact that $|G|$ is fixed.
\end{proof}

Next, we first derive the non-asymptotic Frobenius norm bound with explicit constants for a generic HAC estimator of the sample mean that holds uniformly over a class of distributions. We focus on the $p$-dimensional centered stochastic process $(V_t)_{t\in\Z}$ and put
\begin{equation*}
\Xi = \sum_{k\in\Z}\Gamma_{k}\qquad \text{and}\qquad \tilde \Xi = \sum_{|k|<T}K\left(\frac{k}{M_T}\right)\tilde\Gamma_{k},
\end{equation*}
where $\Gamma_{k}=\E[V_tV_{t+k}^\top]$ and $\tilde\Gamma_{k}=\frac{1}{T}\sum_{t=1}^{T-k}V_tV_{t+k}^\top.$ Put also $\Gamma = (\Gamma_{k})_{k\in\Z}$ and let $\langle.,.\rangle$ be the Frobenius inner product with corresponding Frobenius norm $\|.\|$. The following assumption describes the relevant class of distributions and kernel functions.
\begin{assumption}\label{as:HAC}
	Suppose that (i) $K:\R\to[-1,1]$ is a Riemann integrable function such that $K(0)=1$; (ii) there exists some $\varepsilon,\varsigma>0$ such that $|K(0) - K(x)| \leq L|x|^\varsigma$ for all $|x|<\varepsilon$; (iii) $(V_t)_{t\in\Z}$ is fourth-order stationary; (iv) $\Gamma\in \mathcal{G}(\varsigma,D_1,D_2)$, where
	\begin{equation*}
	\mathcal{G}(\varsigma,D_1,D_2) = \left\{\sum_{k\in\Z}|k|^\varsigma\|\Gamma_{k}\|\leq D_1,\quad \sup_{k\in\Z}\sum_{l\in\Z}\sum_{t\in\Z}\sum_{j,h\in [p]}|\Cov(V_{0,j}V_{k,h},V_{t,j}V_{t+l,h})| \leq D_2 \right\}
	\end{equation*}
	for some $D_1,D_2>0$.
\end{assumption}
Condition (ii) describes the smoothness (or order) of the kernel in the neighborhood of zero. $\varsigma=1$ for the Bartlett kernel and $\varsigma=2$ for the Parzen, Tukey-Hanning, and Quadratic spectral kernels, see \cite{andrews1991heteroskedasticity}. Since the bias of the HAC estimator is limited by the order of the kernel, it is typically not recommended to use the Bartlett kernel in practice. Higher-order kernels with $\varsigma>2$ do not ensure the positive definiteness of the HAC estimator and require additional spectral regularization, see \cite{politis2011higher}. Condition (iv) describes the class of autocovariances that vanish rapidly enough. Note that if (iv) holds for some $\bar\varsigma$, then it also holds for every $\varsigma<\bar\varsigma$ and that if (ii) holds for some $\tilde\varsigma>\varsigma$, then it also holds for $\tilde\varsigma=\varsigma$. The covariance condition in (iv) can be justified under more primitive moment and summability conditions imposed on $L_1$-mixingale/$\tau$-mixing coefficients, see Proposition~\ref{lemma:covariances} and \cite{andrews1991heteroskedasticity}, Lemma 1. The following result gives a nonasymptotic risk bound uniformly over the class $\mathcal{G}$ and corresponds to the asymptotic convergence rates for the spectral density evaluated at zero derived in \cite{parzen1957consistent}.

\begin{proposition}\label{prop:HAC_true}
	Suppose that Assumption~\ref{as:HAC} is satisfied. Then
	\begin{equation*}
	\sup_{\Gamma\in\mathcal{G}(\varsigma,D_1,D_2)}\E\|\tilde\Xi - \Xi\|^2 \leq C_1\frac{M_T}{T} + C_2M_T^{-2\varsigma} + C_3T^{-2(\varsigma\wedge1)},
	\end{equation*}
	where $C_1 = D_2\left(\int|K(u)|\dx u + o(1)\right)$, $C_2 = 2\left(D_1L + \frac{2D_1}{\varepsilon^\varsigma}\right)^2$, and $C_3 = 2D_1^2$.
\end{proposition}
\begin{proof}	
	By the triangle inequality, under Assumption~\ref{as:HAC} (i)
	\begin{equation*}
	\begin{aligned}
	\|\E[\tilde\Xi] - \Xi\| & = \left\|\sum_{|k|<T}K\left(\frac{k}{M_T}\right)\frac{T-k}{T}\Gamma_{k} - \sum_{k\in\Z}\Gamma_{k}\right\| \\
	& \leq \sum_{|k|< T}\left|K\left(\frac{k}{M_T}\right) - K(0)\right|\|\Gamma_{k}\| + \frac{1}{T}\sum_{|k|<T}|k|\|\Gamma_{k}\| + \sum_{|k|\geq T}\|\Gamma_{k}\| \\
	& \triangleq I_T + II_T + III_T.
	\end{aligned}
	\end{equation*}
	For the first term, we obtain 
	\begin{equation*}
	\begin{aligned}
	I_T & =  \sum_{|k|< \varepsilon M_T}\left|K(0) - K\left(\frac{k}{M_T}\right)\right|\|\Gamma_{k}\| + \sum_{\varepsilon M_T\leq |k|< T}\left|K\left(\frac{k}{M_T}\right) - K(0)\right|\|\Gamma_{k}\| \\
	& \leq LM_T^{-\varsigma}\sum_{|k|<\varepsilon M_T}|k|^\varsigma\|\Gamma_{k}\| + 2\sum_{\varepsilon M_T\leq |k|<T}\|\Gamma_{k}\| \\
	& \leq \frac{D_1L}{M_T^\varsigma} + \frac{2}{\varepsilon^\varsigma M_T^\varsigma}\sum_{\varepsilon M_T\leq |k|<T}|k|^\varsigma\|\Gamma_{k}\| \\
	& \leq \frac{D_1L}{M_T^\varsigma} + \frac{2D_1}{\varepsilon^\varsigma M_T^\varsigma},
	\end{aligned}
	\end{equation*}
	where the second sum is defined to be zero if $T\leq \varepsilon M_T$, the second line follows under Assumption~\ref{as:HAC} (i)-(ii) and the last two under Assumption~\ref{as:HAC} (iii). Next, if $\varsigma\geq 1$,
	\begin{equation*}
	\begin{aligned}
	\sum_{|k|<T}|k|\|\Gamma_{k}\| \leq \sum_{|k|<T}|k|^\varsigma\|\Gamma_{k}\|,
	\end{aligned}
	\end{equation*}
	while if $\varsigma\in(0,1)$
	\begin{equation*}
	\sum_{|k|<T}|k|\|\Gamma_{k}\| \leq T^{1-\varsigma}\sum_{|k|<T}|k|^\varsigma\|\Gamma_{k}\|.
	\end{equation*}
	Therefore, since $\sum_{|k|\geq T}\|\Gamma_{k}\| \leq T^{-\varsigma}\sum_{|k|\geq T}|k|^\varsigma\|\Gamma_{k}\|$, under Assumption~\ref{as:HAC} (iv)
	\begin{equation*}
	\begin{aligned}
	II_T + III_T &\leq \begin{cases}
	\frac{D_1}{T} & \varsigma\geq 1 \\
	\frac{D_1}{T^\varsigma} & \varsigma\in(0,1)
	\end{cases} \\
	& = \frac{D_1}{T^{\varsigma\wedge 1}}.
	\end{aligned}
	\end{equation*}
	This shows that
	\begin{equation}\label{eq:hac_1}
	\|\E[\tilde\Xi] - \Xi\| \leq  \frac{D_1L}{M_T^\varsigma} + \frac{2D_1}{\varepsilon^\varsigma M_T^\varsigma} + \frac{D_1}{T^{\varsigma\wedge1}}.
	\end{equation}
	
	\noindent	Next, under Assumption~\ref{as:HAC} (i)
	\begin{equation*}
	\begin{aligned}
	\E\|\tilde\Xi - \E[\tilde\Xi]\|^2 & = \sum_{|k|<T}\sum_{|l|<T}K\left(\frac{k}{M_T}\right)K\left(\frac{l}{M_T}\right)\E\left\langle\tilde\Gamma_k - \E\tilde\Gamma_k, \tilde\Gamma_l - \E\tilde\Gamma_l\right\rangle \\
	& \leq \sum_{|k|<T}\left|K\left(\frac{k}{M_T}\right)\right| \sup_{|k|<T}\sum_{|l|<T}\left|\E\left\langle\tilde\Gamma_k - \E\tilde\Gamma_k, \tilde\Gamma_l - \E\tilde\Gamma_l\right\rangle\right|,
	\end{aligned}
	\end{equation*}
	where under Assumptions~\ref{as:HAC} (iii)
	\begin{equation*}
	\begin{aligned}
	T\left|\E\left\langle\tilde\Gamma_k - \E\tilde\Gamma_k, \tilde\Gamma_l - \E\tilde\Gamma_l\right\rangle\right| & \leq \frac{1}{T}\sum_{t=1}^{T-k}\sum_{r=1}^{T-l}\sum_{j,h\in [p]}\left|\Cov(V_{t,j}V_{t+k,h},V_{r,j}V_{r+l,h})\right|\\
	& \leq \sum_{t\in\Z}\sum_{j,h\in [p]}|\Cov(V_{0,j}V_{k,h},V_{t,j}V_{t+l,h})|.
	\end{aligned}
	\end{equation*}
	Therefore, under Assumptions~\ref{as:HAC} (i), (iv)
	\begin{equation}\label{eq:hac_2}
	\begin{aligned}
	\E\|\tilde\Xi - \E[\tilde\Xi]\|^2 \leq M_T\left(\int|K(u)|\dx u  + o(1)\right)\frac{D_2}{T} .
	\end{aligned}
	\end{equation}
	The result follows from combining estimates in equations (\ref{eq:hac_1}) and (\ref{eq:hac_2}).
\end{proof}

\clearpage

{\scriptsize
	\begin{longtable}{rlrrr}
		\hline
		& News topic & Meta topic & LASSO & sg-LASSO \\ 
		\hline
		1 & Accounting & Asset Managers \& I-Banks &  &  \\ 
		2 & Acquired investment banks & Asset Managers \& I-Banks & $\checkmark$ & $\checkmark$ \\ 
		3 & Activists & Activism/Language &  &  \\ 
		4 & Aerospace/defense & Trans/Retail/Local Politics & $\checkmark$ & $\checkmark$ \\ 
		5 & Agreement reached & Negotiations &  &  \\ 
		6 & Agriculture & Oil \& Mining &  &  \\ 
		7 & Airlines & Trans/Retail/Local Politics &  &  \\ 
		8 & Announce plan & Activism/Language &  &  \\ 
		9 & Arts & Social/Cultural &  &  \\ 
		10 & Automotive & Trans/Retail/Local Politics &  &  \\ 
		11 & Bank loans & Banks &  &  \\ 
		12 & Bankruptcy & Buyouts \& Bankruptcy &  &  \\ 
		13 & Bear/bull market & Financial Markets & $\checkmark$ & $\checkmark$ \\ 
		14 & Biology/chemistry/physics & Science/Language &  &  \\ 
		15 & Bond yields & Financial Markets &  &  \\ 
		16 & Broadcasting & Entertainment & $\checkmark$ & $\checkmark$ \\ 
		17 & Buffett & Activism/Language &  &  \\ 
		18 & Bush/Obama/Trump & Leaders &  &  \\ 
		19 & C-suite & Management &  &  \\ 
		20 & Cable & Industry &  &  \\ 
		21 & California & Trans/Retail/Local Politics &  &  \\ 
		22 & Canada/South Africa & International Affairs &  &  \\ 
		23 & Casinos & Industry & $\checkmark$ &  \\ 
		24 & Challenges & Challenges &  &  \\ 
		25 & Changes & Challenges &  &  \\ 
		26 & Chemicals/paper & Industry &  &  \\ 
		27 & China & International Affairs &  &  \\ 
		28 & Clintons & Leaders &  &  \\ 
		29 & Committees & Negotiations &  &  \\ 
		30 & Commodities & Financial Markets &  &  \\ 
		31 & Company spokesperson & Negotiations &  &  \\ 
		32 & Competition & Industry &  &  \\ 
		33 & Computers & Technology &  &  \\ 
		34 & Connecticut & Management &  &  \\ 
		35 & Control stakes & Buyouts \& Bankruptcy & $\checkmark$ &  \\ 
		36 & Convertible/preferred & Buyouts \& Bankruptcy &  &  \\ 
		37 & Corporate governance & Buyouts \& Bankruptcy & $\checkmark$ & $\checkmark$ \\ 
		38 & Corrections/amplifications & Activism/Language &  &  \\ 
		39 & Couriers & Industry &  &  \\ 
		40 & Courts & Courts & $\checkmark$ &  \\ 
		41 & Credit cards & Industry &  &  \\ 
		42 & Credit ratings & Banks & $\checkmark$ & $\checkmark$ \\ 
		43 & Cultural life & Social/Cultural &  &  \\ 
		44 & Currencies/metals & Financial Markets &  &  \\ 
		45 & Disease & Trans/Retail/Local Politics & $\checkmark$ & $\checkmark$ \\ 
		46 & Drexel & Buyouts \& Bankruptcy & $\checkmark$ &  \\ 
		47 & Earnings & Corporate Earnings &  &  \\ 
		48 & Earnings forecasts & Corporate Earnings &  &  \\ 
		49 & Earnings losses & Corporate Earnings & $\checkmark$ &  \\ 
		50 & Economic growth & Economic Growth &  &  \\ 
		51 & Economic ideology & Social/Cultural &  &  \\ 
		52 & Elections & Leaders & $\checkmark$ & $\checkmark$ \\ 
		53 & Electronics & Technology & $\checkmark$ & $\checkmark$ \\ 
		54 & Environment & Government & $\checkmark$ &  \\ 
		55 & European politics & Leaders & $\checkmark$ &  \\ 
		56 & European sovereign debt & Economic Growth & $\checkmark$ & $\checkmark$ \\ 
		57 & Exchanges/composites & Financial Markets &  &  \\ 
		58 & Executive pay & Labor/income &  &  \\ 
		59 & Fast food & Industry & $\checkmark$ &  \\ 
		60 & Federal Reserve & Economic Growth & $\checkmark$ & $\checkmark$ \\ 
		61 & Fees & Labor/income &  &  \\ 
		62 & Financial crisis & Banks & $\checkmark$ & $\checkmark$ \\ 
		63 & Financial reports & Corporate Earnings &  &  \\ 
		64 & Foods/consumer goods & Industry &  &  \\ 
		65 & France/Italy & International Affairs &  &  \\ 
		66 & Futures/indices & Activism/Language &  &  \\ 
		67 & Gender issues & Social/Cultural &  &  \\ 
		68 & Germany & International Affairs & $\checkmark$ & $\checkmark$ \\ 
		69 & Government budgets & Labor/income &  &  \\ 
		70 & Health insurance & Labor/income &  &  \\ 
		71 & Humor/language & Social/Cultural &  &  \\ 
		72 & Immigration & Social/Cultural &  &  \\ 
		73 & Indictments & Courts &  &  \\ 
		74 & Insurance & Industry &  &  \\ 
		75 & International exchanges & Financial Markets &  &  \\ 
		76 & Internet & Technology &  &  \\ 
		77 & Investment banking & Asset Managers \& I-Banks &  &  \\ 
		78 & IPOs & Financial Markets &  &  \\ 
		79 & Iraq & Terrorism/Mideast & $\checkmark$ &  \\ 
		80 & Japan & International Affairs & $\checkmark$ &  \\ 
		81 & Job cuts & Labor/income &  &  \\ 
		82 & Justice Department & Courts &  &  \\ 
		83 & Key role & Challenges &  &  \\ 
		84 & Latin America & International Affairs & $\checkmark$ & $\checkmark$ \\ 
		85 & Lawsuits & Courts &  &  \\ 
		86 & Long/short term & Challenges &  &  \\ 
		87 & Luxury/beverages & Industry &  &  \\ 
		88 & M\&A & Buyouts \& Bankruptcy & $\checkmark$ &  \\ 
		89 & Machinery & Oil \& Mining &  &  \\ 
		90 & Macroeconomic data & Economic Growth &  &  \\ 
		91 & Major concerns & Activism/Language &  &  \\ 
		92 & Management changes & Management & $\checkmark$ &  \\ 
		93 & Marketing & Entertainment & $\checkmark$ &  \\ 
		94 & Mexico & Activism/Language & $\checkmark$ &  \\ 
		95 & Microchips & Technology &  &  \\ 
		96 & Mid-level executives & Management &  &  \\ 
		97 & Mid-size cities & Trans/Retail/Local Politics &  &  \\ 
		98 & Middle east & Terrorism/Mideast & $\checkmark$ &  \\ 
		99 & Mining & Oil \& Mining &  &  \\ 
		100 & Mobile devices & Technology &  &  \\ 
		101 & Mortgages & Banks & $\checkmark$ &  \\ 
		102 & Movie industry & Entertainment & $\checkmark$ &  \\ 
		103 & Music industry & Entertainment &  &  \\ 
		104 & Mutual funds & Asset Managers \& I-Banks &  &  \\ 
		105 & NASD & Asset Managers \& I-Banks & $\checkmark$ & $\checkmark$ \\ 
		106 & National security & Government & $\checkmark$ &  \\ 
		107 & Natural disasters & Trans/Retail/Local Politics & $\checkmark$ & $\checkmark$ \\ 
		108 & Negotiations & Negotiations &  &  \\ 
		109 & News conference & Negotiations &  &  \\ 
		110 & Nonperforming loans & Banks & $\checkmark$ &  \\ 
		111 & Nuclear/North Korea & Terrorism/Mideast & $\checkmark$ &  \\ 
		112 & NY politics & Trans/Retail/Local Politics & $\checkmark$ &  \\ 
		113 & Oil drilling & Oil \& Mining &  &  \\ 
		114 & Oil market & Oil \& Mining & $\checkmark$ &  \\ 
		115 & Optimism & Economic Growth &  &  \\ 
		116 & Options/VIX & Financial Markets & $\checkmark$ & $\checkmark$ \\ 
		117 & Pensions & Labor/income &  &  \\ 
		118 & People familiar & Negotiations &  &  \\ 
		119 & Pharma & Trans/Retail/Local Politics &  &  \\ 
		120 & Phone companies & Technology &  &  \\ 
		121 & Police/crime & Trans/Retail/Local Politics &  &  \\ 
		122 & Political contributions & Government &  &  \\ 
		123 & Positive sentiment & Social/Cultural &  &  \\ 
		124 & Private equity/hedge funds & Asset Managers \& I-Banks &  &  \\ 
		125 & Private/public sector & Government &  &  \\ 
		126 & Problems & Challenges &  &  \\ 
		127 & Product prices & Economic Growth &  &  \\ 
		128 & Profits & Corporate Earnings & $\checkmark$ & $\checkmark$ \\ 
		129 & Programs/initiatives & Science/Language &  &  \\ 
		130 & Publishing & Entertainment & $\checkmark$ &  \\ 
		131 & Rail/trucking/shipping & Trans/Retail/Local Politics &  &  \\ 
		132 & Reagan & Leaders &  &  \\ 
		133 & Real estate & Buyouts \& Bankruptcy &  &  \\ 
		134 & Recession & Economic Growth & $\checkmark$ & $\checkmark$ \\ 
		135 & Record high & Economic Growth &  &  \\ 
		136 & Regulation & Government &  &  \\ 
		137 & Rental properties & Trans/Retail/Local Politics & $\checkmark$ &  \\ 
		138 & Research & Science/Language &  &  \\ 
		139 & Restraint & Negotiations &  &  \\ 
		140 & Retail & Trans/Retail/Local Politics &  &  \\ 
		141 & Revenue growth & Industry &  &  \\ 
		142 & Revised estimate & Corporate Earnings &  &  \\ 
		143 & Russia & International Affairs & $\checkmark$ &  \\ 
		144 & Safety administrations & Government &  &  \\ 
		145 & Sales call & Social/Cultural &  &  \\ 
		146 & Savings \& loans & Banks & $\checkmark$ & $\checkmark$ \\ 
		147 & Scenario analysis & Science/Language &  &  \\ 
		148 & Schools & Social/Cultural &  &  \\ 
		149 & SEC & Buyouts \& Bankruptcy &  &  \\ 
		150 & Share payouts & Financial Markets &  &  \\ 
		151 & Short sales & Financial Markets & $\checkmark$ &  \\ 
		152 & Size & Science/Language &  &  \\ 
		153 & Small business & Industry &  &  \\ 
		154 & Small caps & Financial Markets & $\checkmark$ & $\checkmark$ \\ 
		155 & Small changes & Corporate Earnings &  &  \\ 
		156 & Small possibility & Challenges &  &  \\ 
		157 & Soft drinks & Industry &  &  \\ 
		158 & Software & Technology & $\checkmark$ &  \\ 
		159 & Southeast Asia & International Affairs &  &  \\ 
		160 & Space program & Science/Language &  &  \\ 
		161 & Spring/summer & Challenges &  &  \\ 
		162 & State politics & Government &  &  \\ 
		163 & Steel & Oil \& Mining &  &  \\ 
		164 & Subsidiaries & Industry &  &  \\ 
		165 & Systems & Science/Language &  &  \\ 
		166 & Takeovers & Buyouts \& Bankruptcy & $\checkmark$ &  \\ 
		167 & Taxes & Labor/income & $\checkmark$ &  \\ 
		168 & Terrorism & Terrorism/Mideast &  &  \\ 
		169 & Tobacco & Industry &  &  \\ 
		170 & Trade agreements & International Affairs & $\checkmark$ & $\checkmark$ \\ 
		171 & Trading activity & Financial Markets &  &  \\ 
		172 & Treasury bonds & Financial Markets &  &  \\ 
		173 & UK & International Affairs &  &  \\ 
		174 & Unions & Labor/income & $\checkmark$ & $\checkmark$ \\ 
		175 & US defense & Trans/Retail/Local Politics & $\checkmark$ &  \\ 
		176 & US Senate & Leaders &  &  \\ 
		177 & Utilities & Government & $\checkmark$ &  \\ 
		178 & Venture capital & Industry &  &  \\ 
		179 & Watchdogs & Government &  &  \\ 
		180 & Wide range & Science/Language &  &  \\ 
		\hline
		\caption{\small News series -- The column {\it News topic} are the news series topics, column {\it Meta topic} are meta topics/groups of news series. Columns {\it LASSO} and {\it sg-LASSO} reports whether the series was selected ($\checkmark$) or not by the respective initial estimator. \label{tab:selected}	}
	\end{longtable}
}

\end{document}